\documentclass[12pt, a4paper]{article}

\usepackage{amsfonts,amsmath,amssymb,amsthm,latexsym}
\usepackage{times,mathptm}
\usepackage[all]{xy}

\newtheorem{proposition}{Proposition}[section]
\newtheorem{lemma}[proposition]{Lemma}
\newtheorem{corollary}[proposition]{Corollary}
\newtheorem{theorem}[proposition]{Theorem}

\theoremstyle{definition}
\newtheorem{definition}[proposition]{Definition}

\theoremstyle{remark}

\newcommand{\thlabel}[1]{\label{th:#1}}
\newcommand{\thref}[1]{Theorem~\ref{th:#1}}
\newcommand{\selabel}[1]{\label{se:#1}}
\newcommand{\seref}[1]{Section~\ref{se:#1}}
\newcommand{\lelabel}[1]{\label{le:#1}}

\newcommand{\prlabel}[1]{\label{pr:#1}}
\newcommand{\prref}[1]{Proposition~\ref{pr:#1}}
\newcommand{\colabel}[1]{\label{co:#1}}
\newcommand{\coref}[1]{Corollary~\ref{co:#1}}

\newcommand{\delabel}[1]{\label{de:#1}}
\newcommand{\deref}[1]{Definition~\ref{de:#1}}

\author{K. Kanakoglou, Physics Department, Aristotle University \\
of Thessaloniki, Thessaloniki, 54124, GREECE, \\
kanakoglou@hotmail.com \and C. Daskaloyannis, Mathematics
Department, Aristotle University \\
of Thessaloniki, Thessaloniki, 54124, GREECE, \\
 daskalo@math.auth.gr}
\title{Paraboson quotients. A braided look at Green's
ansatz and a generalization \footnote{J. Math. Phys.
v.\textbf{48}, 113516, (2007)}}

\begin{document}

\maketitle

\begin{abstract}
Bosons and Parabosons are described as associative superalgebras,
with an infinite number of odd generators. Bosons are shown to be
a quotient superalgebra of Parabosons, establishing thus an even
algebra epimorphism which is an immediate link between their
simple modules. Parabosons are shown to be a super-Hopf algebra.
The super-Hopf algebraic structure of Parabosons, combined with
the projection epimorphism previously stated, provides us with a
braided interpretation of the Green's ansatz device and of the
parabosonic Fock-like representations. This braided interpretation
combined with an old problem leads to the construction of a
straightforward generalization of Green's ansatz.
\end{abstract}

\section{Introduction} \selabel{intro}

For a quantum system with a Hamiltonian of the form $H(p_{i},
q_{i})$ with possibly infinite degrees of freedom $i=1, 2, \ldots
\ $, the canonical variables $p_{i}, q_{i}$ are considered to be
generators of a unital associative non-commutative algebra,
described in terms of generators and relations by (we have set
$\hbar = 1$):
\begin{equation} \label{CCR}
[q_{i}, p_{j}] = i \ \delta_{ij} I \ \ \ \ \ \ \ \ \ \ [q_{i},
q_{j}] = [p_{i}, p_{j}] = 0
\end{equation}
$I$ is of course the unity of the algebra, $i,j = 1, 2, \ldots \ $
and $[x, y]$ stands for $xy-yx$. Relations \eqref{CCR} are known
in the physical literature as the Weyl algebra, or the
Heisenberg-Weyl algebra or more commonly as the Canonical
Commutation Relations often abbreviated as CCR.  For  technical
reasons it is common to use -instead of the variables $p_{i},
q_{i}$- the linear combinations:
\begin{equation} \label{creation-destruction op}
b_{j}^{+} = \frac{1}{\sqrt{2}}(q_{j} - ip_{j}) \ \ \ \ \ \ \ \
b_{j}^{-} = \frac{1}{\sqrt{2}}(q_{j} + ip_{j})
\end{equation}
for $j=1, 2, \ldots \ $ in terms of which \eqref{CCR} become:
\begin{equation} \label{CCRbose}
[b_{i}^{-}, b_{j}^{+}] =  \delta_{ij} I \ \ \ \ \ \ \ \ \ \
[b_{i}^{-}, b_{j}^{-}] = [b_{i}^{+}, b_{j}^{+}] = 0
\end{equation}
for $i,j = 1, 2, \ldots \ $. These latter relations are usually
called the bosonic algebra, and it is well known that this algebra
is the starting point for the second quantisation problem (in the
case of the free field theory) leading to the Bose-Einstein
statistics.

Parabosonic algebra grew out of the desire to generalize the
second quantization method -in the case of the the free field- in
a way permitting more general kind of statistics than the
Bose-Einstein statistics. It was formally introduced in terms of
generators and relations by Green in \cite{Green}, Greenberg and
Messiah in \cite{GreeMe} and Volkov in \cite{Vo}. It is generated
by the generators $B_{i}^{+}, B_{i}^{-}$ subject to the relations:
\begin{equation} \label{eq:PCR}
 \big[ \{ B_{i}^{\xi},  B_{j}^{\eta}\}, B_{k}^{\epsilon}  \big] -
 (\epsilon - \eta)\delta_{jk}B_{i}^{\xi} - (\epsilon -
 \xi)\delta_{ik}B_{j}^{\eta} = 0
\end{equation}
for all values of $i, j = 1, 2, ... \ $ and $ \xi, \eta, \epsilon
= \pm 1 \ $. The field theory statistics stemming from such an
algebra, has been known with the name ``parastatistics" and is
still a wide open subject of research. (see \cite{OhKa} and
references therein). \\

This paper consists logically of four parts: The first part
consists of \seref{1}. We state the definitions and derive the
$\mathbb{Z}_{2}$-grading of bosonic and parabosonic algebras (in
infinite degrees of freedom). Although the bosonic algebra is
usually considered to be a quotient of the UEA of the Heisenberg
Lie algebra with it's generators thus being even elements, we
adopt here a totaly different approach: We consider the bosonic
algebra to be an associative $\mathbb{Z}_2$-graded algebra
(associative superalgebra) with it's generators being odd
elements. This approach is reminiscent -although actually
different- of the antisymmetric Clifford-Weyl algebra (see for
example \cite{Pal3}). It permits us to express the bosonic algebra
as a quotient superalgebra of the parabosonic superalgebra. The
projection epimorphism from parabosons to bosons creates thus an
immediate link between irreducible representations of the bosonic
algebra and irreducible representations of the parabosonic
algebra.

The second part of the paper consists of \seref{11} and \seref{2}.
In \seref{11} the notions of $\mathbb{Z}_{2}$-graded algebra,
$\mathbb{Z}_{2}$-graded modules and $\mathbb{Z}_{2}$-graded tensor
products \cite{Che}, are reviewed as special examples of the more
general and modern notions of $\mathbb{G}$-module algebras
($\mathbb{G}$: a finite abelian group) and of braiding in monoidal
categories \cite{Mon, Maj2, Maj3}. The role of the non-trivial
quasitriangular structure of the $\mathbb{CZ}_{2}$ group Hopf
algebra or equivalently the role of the braiding of the category
${}_{\mathbb{CZ}_{2}}\mathcal{M}$ of $\mathbb{CZ}_{2}$ modules is
emphasized in order to make clear that similar constructions can
be carried out in the same way for more complicated gradings (like
for example $\mathbb{Z}_{2} \times \mathbb{Z}_{2}$-grading). In
\seref{2} the notion of the super-Hopf algebra is reviewed and the
super-Hopf algebraic structure of the parabosonic algebra is
established. Note that the proof, does not make use of the well
known \cite{Pal} Lie superalgebraic structure of the parabosonic
algebra (for the case of the finite degrees of freedom). We
conclude the paragraph with an indication of how the super-Hopf
algebraic structure can be used to form multiple tensor products
of (braided) representations of the original super-Hopf algebra.

The third part of the paper consists of \seref{Gr} and
\seref{brainterpGrAns} and deals with the representation theory of
the bosonic and parabosonic algebras. After a review of well known
results in \seref{Gr}, we apply in \seref{brainterpGrAns} the
techniques developed in the previous sections to the case of the
braided representations of the parabosonic super-Hopf algebra. The
Green ansatz algebras are shown to be isomorphic to braided tensor
product algebras constructed from multiple copies of the bosonic
algebra. The $\mathbb{Z}_{2}$-grading of the bosonic algebra,
established in \seref{1} plays an essential role in the proof.
Furthermore, the Fock-like representations of the parabosonic
algebra -which were classified by Greenberg in \cite{GreeMe}- are
shown to be braided ($\mathbb{Z}_{2}$-graded), irreducible
(simple) submodules of the multiple tensor product module of the
first Fock-like representation.

Finally, the last part of this paper consists of
\seref{Generalizations} alone. We combine the isomorphism stated
in \seref{brainterpGrAns} with some partial solutions of an old
problem of Ohnuki and Kamefuchi (see \cite{OhKa}): the
construction of \emph{self contained sets of commutation
relations} or \emph{commutation relations specific to given order
$p$}. As a result we construct a straightforward generalization of
the Green's ansatz device. Furthermore our construction indicates
how a whole family of similar generalizations can be constructed,
provided we have (although we don't) the general solution of
Ohnuki's and Kamefuchi's problem. Speaking about generalizations
of the Green's ansatz, one should also see \cite{Ques} for a
different approach.

 In what follows, all vector spaces and algebras and all
tensor products will be considered over the field of complex
numbers. Whenever the symbol $i$ enters a formula in another place
than an index, it always denotes the imaginary unit $i^{2} = -1$.
Furthermore, whenever formulas from physics enter the text, we use
the traditional convention: $\hbar = m = \omega = 1$. Finally, the
Sweedler's notation for the comultiplication is freely used
throughout the text.

\section{Bosons, Parabosons and superalgebra quotients} \selabel{1}

The parabosonic algebra, was originally defined in terms of
generators and relations by Green \cite{Green} and
Greenberg-Messiah \cite{GreeMe}. We begin with restating their
definition: \\
 Let us consider the vector space $V_{X}$ freely
generated by the elements: $X_{i}^{+}, X_{j}^{-}$, $i,j=1, 2,
...$. Let $T(V_{X})$ denote the tensor algebra of $V_{X}$.
$T(V_{X})$ is - up to isomorphism - the free algebra generated by
the elements of the basis. In $T(V_{X})$ we consider the two-sided
ideal $I_{B}$, generated by the following elements:
\begin{equation} \label{eq:bdef}
\begin{array}{ccccc}
 \big[ X_{i}^{-}, X_{j}^{+} \big] - \delta_{ij}I_{X} & , & \big[ X_{i}^{-}, X_{j}^{-} \big]
 & , & \big[ X_{i}^{+},  X_{j}^{+} \big]
   \\
\end{array}
\end{equation}
for all values of $i,j=1, 2, ... \ $. $ \ I_{X}$ is the unity of
the tensor algebra. We now have the following:
\begin{definition} \delabel{bosons}
The bosonic algebra  $B$ is the quotient algebra of the tensor
algebra $T(V_{X})$ with the ideal $I_{B}$:
$$
B = T(V_{X}) / I_{B}
$$
\end{definition}
We denote by $\pi_{B}: T(V_{X}) \rightarrow B \ $ the canonical
projection. The elements $X_{i}^{+}$, $X_{j}^{-}$, $I_{X}$, where
$i,j=1, 2, ... \ $ and $I_{X}$ is the unity of the tensor algebra,
are the generators of the tensor algebra $T(V_{X})$. The elements
$\pi_{B}(X_{i}^{+}), \ \pi_{B}(X_{j}^{-}), \ \pi_{B}(I_{X}) \ $,
for all values $ \ i,j=1, 2, ...$ are a set of generators of the
bosonic algebra $B$, and they will be denoted by $b_{i}^{+},
b_{j}^{-}, I$ for $i,j=1, 2, ...$ respectively, from now on.
$\pi_{B}(I_{X}) = I$ is the unity of the bosonic algebra. In the
case of the finite degrees of freedom $i, j = 1, 2, ... m$ we have
the bosonic algebra in $2m$ generators ($m$ boson algebra) which
we will denote by $B^{(m)}$ from now on.

Returning again in the tensor algebra $T(V_{X})$ we consider the
two-sided ideal $I_{P_{B}}$, generated by the following elements:
\begin{equation} \label{eq:pbdef}
 \big[ \{ X_{i}^{\xi},  X_{j}^{\eta}\}, X_{k}^{\epsilon}  \big] -
 (\epsilon - \eta)\delta_{jk}X_{i}^{\xi} - (\epsilon - \xi)\delta_{ik}X_{j}^{\eta}
\end{equation}
respectively, for all values of $\xi, \eta, \epsilon = \pm 1$ and
$i,j=1, 2, ... \ $. $ \ I_{X}$ is the unity of the tensor algebra.
We now have the following:
\begin{definition} \delabel{parabosonsbosons}
The parabosonic algebra in $P_{B}$  is the quotient algebra of the
tensor algebra $T(V_{X})$ of $V_{X}$ with the ideal $I_{P_{B}}$:
$$
P_{B} = T(V_{X}) / I_{P_{B}}
$$
\end{definition}
We denote by $\pi_{P_{B}}: T(V_{X}) \rightarrow P_{B} \ $ the
canonical projection. The elements $ \pi_{P_{B}}(X_{i}^{+})$, $
\pi_{P_{B}}(X_{j}^{-})$,  $ \pi_{P_{B}}(I_{X}) \ $, for all values
$ \ i,j=1,...$ are a set of generators of the parabosonic algebra
$P_{B}$, and they will be denoted by $B_{i}^{+}, B_{j}^{-}, I$ for
$i,j=1, 2, ...$ respectively, from now on. $\pi_{P_{B}}(I_{X}) =
I$ is the unity of the parabosonic algebra. In the case of the
finite degrees of freedom $i, j = 1, 2, ... m$ we have the
parabosonic algebra in $2m$ generators ($m$ paraboson algebra)
which we will denote by $P_{B}^{(m)}$ from now on.

Based on the above definitions we prove now the following
proposition which clarifies the relationship between bosonic and
parabosonic algebras:
\begin{proposition} \prlabel{parabosonstobosons}
The parabosonic algebra $P_{B}$ and the bosonic algebra $B$ are
both $\mathbb{Z}_{2}$-graded algebras with their generators
$B_{i}^{\pm}$ and $b_{i}^{\pm}$ respectively, $i,j=1, 2, ...$,
being odd elements. The bosonic algebra $B$ is a quotient algebra
of the parabosonic algebra $P_{B}$. The ``replacement" map $\phi:
P_{B} \rightarrow B$ defined by: $\phi(B_{i}^{\pm}) = b_{i}^{\pm}$
is a $\mathbb{Z}_{2}$-graded algebra epimorphism (i.e.: an even
algebra epimorphism).
\end{proposition}
\begin{proof}
It is obvious that the tensor algebra $T(V_{X})$ is a
$\mathbb{Z}_{2}$-graded algebra with the monomials being
homogeneous elements. If $x$ is an arbitrary monomial of the
tensor algebra, then $deg(x) = 0$, namely $x$ is an even element,
if it constitutes of an even number of factors (an even number of
generators of $T(V_{X})$) and $deg(x) = 1$, namely $x$ is an odd
element, if it constitutes of an odd number of factors (an odd
number of generators of $T(V_{X})$). The generators $X_{i}^{+},
X_{j}^{-} \ $, $ \ i,j=1,...,n$ are odd elements in the above
mentioned gradation.
 In view of the above description we can easily conclude that the
$\mathbb{Z}_{2}$-gradation of the tensor algebra is immediately
``transfered" to the algebras $P_{B}$ and $B$: Both ideals
$I_{P_{B}}$ and $I_{B}$ are homogeneous ideals of the tensor
algebra, since they are generated by homogeneous elements of
$T(V_{X})$. Consequently, the projection homomorphisms
$\pi_{P_{B}}$ and $\pi_{B}$ are homogeneous algebra maps of degree
zero, or we can equivalently say that they are even algebra
homomorphisms.
 We can straightforwardly check that the bosons
satisfy the paraboson relations, i.e:
$$
\begin{array}{c}
  \pi_{B}( \big[ \{ X_{i}^{\xi},  X_{j}^{\eta}\},
X_{k}^{\epsilon}  \big] -
 (\epsilon - \eta)\delta_{jk}X_{i}^{\xi} - (\epsilon -
 \xi)\delta_{ik}X_{j}^{\eta} ) = \\
              \\
 = \big[ \{ b_{i}^{\xi},  b_{j}^{\eta}\},
b_{k}^{\epsilon}  \big] -
 (\epsilon - \eta)\delta_{jk}b_{i}^{\xi} - (\epsilon -
 \xi)\delta_{ik}b_{j}^{\eta} = 0 \\
\end{array}
$$
which simply means that: $ker(\pi_{P_{B}}) \subseteq ker(\pi_{B})$
or equivalently: $I_{P_{B}} \subseteq I_{B}$. By the correspodence
theorem for rings, we get that the set $I_{B} / I_{P_{B}} =
\pi_{P_{B}}(I_{B})$ is an homogeneous ideal of the algebra
$P_{B}$, and applying the third isomorphism theorem for rings we
get:
\begin{equation} \label{thirdisomortheor}
P_{B} \Big/ (I_{B} / I_{P_{B}}) = (T(V_{X}) / I_{P_{B}}) \Big/
(I_{B} / I_{P_{B}}) \cong T(V_{X}) \Big/ I_{B} = B
\end{equation}
Thus we have shown that the bosonic algebra $B$ is a quotient
algebra of the parabosonic algebra $P_{B}$. The fact that
$I_{P_{B}} \subseteq I_{B}$ implies that $\pi_{B}$ is uniquely
extended to an even algebra homomorphism $\phi: P_{B} \rightarrow
B$ according to the following (commutative) diagram:
\begin{displaymath}
\xymatrix{T(V_{X}) \ar[rr]^{\pi_{B}}
\ar[dr]_{\pi_{P_{B}}} & & B \\
 & P_{B} \ar@{.>}[ur]_{\exists ! \ \phi} & }
\end{displaymath}
where $\phi$ is completely determined by it's values on the
generators $B_{i}^{\pm}$ of $P_{B}$, i.e.: $\phi(B_{i}^{\pm}) =
b_{i}^{\pm}$. Recalling now that: $ker\phi = I_{B} / I_{P_{B}} =
\pi_{P_{B}}(I_{B})$ and using equation \eqref{thirdisomortheor},
we get the $\mathbb{Z}_{2}$-graded algebra isomorphism:
$$
B \cong P_{B} \Big/ ker\phi
$$
which completes the proof that $\phi$ is an epimorphism of
$\mathbb{Z}_{2}$-graded algebras (or: an even epimorphism).
\end{proof}
Note that $ker\phi$ is exactly the ideal of $P_{B}$ generated by
the elements of the form: $ \ \big[ B_{i}^{-}, B_{j}^{+} \big] -
\delta_{ij}I \ $, $ \ \big[ B_{i}^{-}, B_{j}^{-} \big] \ $, $ \
\big[ B_{i}^{+},  B_{j}^{+} \big] \ $ for all values of $i,j=1, 2,
...$, and $I$ is the unity of the $P_{B}$ algebra. We have an
immediate corollary:
\begin{corollary} \colabel{Parabreprfrombosrepr}
For any vector space $V$, any ${}_{B}V$ module of the bosonic
algebra $B$ immediately gives rise to a ${}_{P_{B}}V$ module of
the parabosonic algebra $P_{B}$ through the replacement map:
\begin{equation} \label{Parabactfrbosact}
B_{i}^{\pm} \cdot x = \phi(B^{\pm}) \cdot x = b_{i}^{\pm} \cdot x
\end{equation}
where $x$ is any element of $ \ V$. Furthermore, if ${}_{B}V$ is
an irreducible representation (a simple module) the fact that
$\phi$ is an epimorphism implies that ${}_{P_{B}}V$ is also an
irreducible representation (a simple module).
\end{corollary}
At this point we should underline the difference between our
approach to the bosonic algebra (or: Weyl algebra or: CCR) as
stated in \deref{parabosonsbosons} and the more mainstream
approach which considers the bosonic algebra to be a quotient
algebra of the universal enveloping algebra of the Heisenberg Lie
algebra: The Heisenberg Lie algebra has -as a complex vector
space- the basis: $ \ b_{i}^{\pm}, c \ $ for all values $i = 1, 2,
... $ and it's Lie algebraic structure is determined by the
relations:
\begin{equation}
[b_{i}^{-}, b_{j}^{+}] =  \delta_{ij} c \ \ \ \ \ \ \ \ \ \
[b_{i}^{-}, b_{j}^{-}] = [b_{i}^{+}, b_{j}^{+}] = [b_{i}^{\pm}, c]
= 0
\end{equation}
If we denote the Heisenberg Lie algebra by $L_{H}$ and it's
universal enveloping algebra by $U(L_{H})$, then we can
immediately see that we have the associative algebra isomorphism:
\begin{equation} \label{WeylquotHeisenb}
B \ \cong \ U(L_{H}) \Big/ <c-I>
\end{equation}
which enables us to consider the bosonic algebra $B$ as a quotient
algebra of the universal enveloping algebra of the Lie algebra
$L_{H}$. By $<c-I>$ we denote the ideal of $U(L_{H})$ generated by
the element $c-I$, where $I$ is the unity of $U(L_{H})$. We should
mention here that the isomorphism stated in equation
\eqref{WeylquotHeisenb} is an algebra isomorphism but not a
$\mathbb{Z}_{2}$-graded algebra isomorphism, since in $U(L_{H})
\Big/ <c-I>$ the generators $b_{i}^{\pm}$ are considered to be
even (ungraded) elements, while according to our approach, in the
bosonic algebra $B$ the same generators are considered to be odd
elements. It is the $\mathbb{Z}_{2}$-grading that makes the whole
difference, and the reason for this point of view will become
clear in the sequel where we are going to discuss tensor products
and representations.

Let us stress here, that in the case of finite degrees of freedom,
our approach as expressed in \prref{parabosonstobosons} combined
with the well known results of \cite{Pal} indicates a Lie
superalgebraic rather than a Lie algebraic interpretation of the
bosonic algebra $B^{(m)}$: In \cite{Pal} Ganchev and Palev have
shown that the parabosonic algebra $P_{B}^{(m)}$ is isomorphic to
the universal enveloping algebra of the orthosymplectic Lie
superalgebra $osp(1/2m)$, with the isomorphism being a
superalgebra isomorphism (or equivalently: a
$\mathbb{Z}_{2}$-graded algebra isomorphism). This result combined
with \prref{parabosonstobosons} says that the $m$ boson algebra is
isomorphic to a quotient of $U(osp(1/2m))$:
\begin{equation} \label{bosquotpar}
B^{(m)} \cong U(osp(1/2m)) \Big/ ker \phi
\end{equation}
where $\phi$ is the replacement map $\phi : P_{B}^{(m)}
\rightarrow B^{(m)}$. We stress here that -unlike equation
\eqref{WeylquotHeisenb}- equation \eqref{bosquotpar} is a
superalgebra isomorphism and not merely an algebra isomorphism.

\section{Superalgebras, quasitriangularity and braiding} \selabel{11}

 The rise of the theory of quasitriangular Hopf algebras from the
mid-80's \cite{Dri} and thereafter and especially the study and
abstraction of their representations (see: \cite{Maj2, Maj3},
\cite{Mon} and references therein), has provided us with a novel
understanding of the notion and the properties of
$\mathbb{G}$-graded algebras,
where $\mathbb{G}$ is a finite abelian group:  \\
 Restricting ourselves to the simplest case
where $\mathbb{G} = \mathbb{Z}_{2}$, we recall that an algebra $A$
being a $\mathbb{Z}_{2}$-graded algebra (in the physics literature
the term superalgebra is also of widespread use) is equivalent to
saying that $A$ is a $\mathbb{CZ}_{2}$-module algebra, via the
$\mathbb{Z}_{2}$-action determined by: $1 \triangleright a = a$
and $g \triangleright a = (-1)^{|a|}a \ $ (for any $a$ homogeneous
in $A$ and $\ |a| \ $ it's degree). We denote by $1, g$ are the
elements of the $\mathbb{Z}_{2}$ group (written multiplicatively).
What we actually mean is that $A$, apart from being an algebra is
also a $\mathbb{CZ}_{2}$-module and at the same time it's
structure maps (i.e.: the multiplication and the unity map which
embeds the field into the center of the algebra) are
$\mathbb{CZ}_{2}$-module maps which is nothing else but
homogeneous linear maps of degree $0$ (i.e.: even linear maps).
Stated more generally, what we are actually saying is that the
$\mathbb{G}$-grading of $A$ can be equivalently described in terms
of a specific action of the finite abelian group $\mathbb{G}$ on
$A$, thus in terms of a specific action of the $\mathbb{CG}$ group
algebra on $A$. This is not something new. In fact such ideas
already appear in works such as \cite{CohMon} and \cite{Stee}.

 What is actually new and has been developed from the $90$'s
and thereafter \cite{Maj3}, \cite{Mon}, is the role of the
quasitriangularity of the $\mathbb{CG}$ group Hopf algebra (for
$\mathbb{G}$ a finite abelian group, see \cite{Scheu1}) and the
role played by the braiding of the $\mathbb{CG}$-modules in the
construction of the tensor products of $\mathbb{G}$-graded
objects. It is well known that for any group $\mathbb{G}$, the
group algebra $\mathbb{CG}$ equipped with the maps:
\begin{equation}
\begin{array}{ccccc}
  \Delta(g) = g \otimes g &  & \varepsilon(g) = 1 &  & S(g) = g^{-1} \\
\end{array}
\end{equation}
for any $g \in \mathbb{G}$, becomes a Hopf algebra. Focusing again
in the special case $\mathbb{G} = \mathbb{Z}_{2}$, we can further
summarize the description of the preceding paragraph, saying that
$A$ is an algebra in the braided monoidal category of
$\mathbb{CZ}_{2}$-modules ${}_{\mathbb{CZ}_{2}}\mathcal{M}$. In
this case the braiding is induced by the non-trivial
quasitriangular structure of the $\mathbb{CZ}_{2}$ Hopf algebra
i.e. by the non-trivial $R$-matrix:
\begin{equation} \label{eq:nontrivRmatrcz2}
R_{Z_{2}} = \frac{1}{2}(1 \otimes 1 + 1 \otimes g + g \otimes 1 -
g \otimes g)
\end{equation}
In the above relation $1, g$ are the elements of the
$\mathbb{Z}_{2}$ group (written multiplicatively).

 We digress here for a moment, to recall
that (see \cite{Maj2, Maj3} or \cite{Mon}) if $(H,R_{H})$ is a
quasitriangular Hopf algebra through the $R$-matrix $R_{H} = \sum
R_{H}^{(1)} \otimes R_{H}^{(2)}$, then the category of modules
${}_{H}\mathcal{M}$ is a braided monoidal category, where the
braiding is given by a natural family of isomorphisms $\Psi_{V,W}:
V \otimes W \cong W \otimes V$, given explicitly by:
\begin{equation} \label{eq:braid}
\Psi_{V,W}(v \otimes w) = \sum (R_{H}^{(2)} \vartriangleright w)
\otimes (R_{H}^{(1)} \vartriangleright v)
\end{equation}
for any $V,W \in obj({}_{H}\mathcal{M})$. By $v,w$ we denote any
 elements of $V,W$ respectively.   \\
Combining eq. \eqref{eq:nontrivRmatrcz2} and \eqref{eq:braid} we
immediately get the braiding in the
${}_{\mathbb{CZ}_{2}}\mathcal{M}$ category:
\begin{equation} \label{symmbraid}
\Psi_{V,W}(v \otimes w) = (-1)^{|v||w|} w \otimes v
\end{equation}
In the above relation $\ |.| \ $ denotes the degree of an
homogeneous element of either $V$ or $W$ (i.e.: $|x|=0$ if $x$ is
an even element and $|x|=1$ if $x$ is an odd element).
 This is obviously a symmetric braiding,
since $\Psi_{V,W} \circ \Psi_{W,V} = Id$, so we actually have a
symmetric monoidal category ${}_{\mathbb{CZ}_{2}}\mathcal{M}$,
rather than a truly braided one.

 The really important thing about the existence of the braiding
\eqref{symmbraid} is that it provides us with an alternative way
of forming tensor products of $\mathbb{Z}_{2}$-graded algebras: If
$A$ and $B$ are superalgebras with multiplications $m_{A}: A
\otimes A \rightarrow A$ and $m_{B}: B \otimes B \rightarrow B$
respectively, then the super vector space $A \otimes B$ (with the
obvious $\mathbb{Z}_{2}$-gradation) equipped with the associative
multiplication
\begin{equation} \label{braidedtenspr}
(m_{A} \otimes m_{B})(Id \otimes \Psi_{B,A} \otimes Id): A \otimes
B \otimes A \otimes B \longrightarrow A \otimes B
\end{equation}
 given equivalently by:
$$
(a \otimes b)(c \otimes d) = (-1)^{|b||c|}ac \otimes bd
$$
for $b,c$ homogeneous in $B, A$ respectively, readily becomes a
superalgebra (or equivalently an algebra in the braided monoidal
category of $\mathbb{CZ}_{2}$-modules
${}_{\mathbb{CZ}_{2}}\mathcal{M}$) which we will denote: $A
\underline{\otimes} B$ and call the braided tensor product algebra
from now on. Let us close this description with two important
remarks: First, we stress that in \eqref{braidedtenspr} both
superalgebras $A$ and $B$ are viewed as $\mathbb{CZ}_{2}$-modules
and as such we have $B \otimes A \cong A \otimes B$ through $b
\otimes c \mapsto (-1)^{|c||b|} c \otimes b$. Second we underline
that the tensor product \eqref{braidedtenspr} had been already
known from the past \cite{Che} but rather as a special possibility
of forming tensor products of superalgebras than as an example of
the more general conceptual framework of the braiding applicable
not only to superalgebras but to any $\mathbb{G}$-graded algebra
($\mathbb{G}$ a finite abelian group) as long as $\mathbb{CG}$ is
equipped with a non-trivial quasitriangular structure or
equivalently \cite{Mon}, \cite{Scheu1}, a bicharacter on
$\mathbb{G}$ is given.

For the case $A=B$ we get the braided tensor algebra $A
\underline{\otimes} A$. Of course we can form longer -braided-
tensor product algebras $A \underline{\otimes} \ldots
\underline{\otimes} A$ between any (finite) number of copies of
$A$. In computing the multiplication in the braided tensor product
algebra $A \underline{\otimes} \ldots \underline{\otimes} A$ the
only thing we have to keep in mind is that each time two copies of
$A$ have to interchange their order, the suitable sign factor (due
to the
braiding \eqref{symmbraid}) has to be inserted. \\
Finally, due to \prref{parabosonstobosons} we immediately check
that the above description fits perfectly the superalgebras
$P_{B}$ and $B$, and offers us a method for constructing the
braided tensor product algebras $P_{B} \underline{\otimes} \ldots
\underline{\otimes} P_{B}$ and $B \underline{\otimes} \ldots
\underline{\otimes} B$ respectively.
\\

A similar discussion applies to the case of
$\mathbb{Z}_{2}$-graded $A$-modules where $A$ is a
$\mathbb{Z}_{2}$-graded algebra. We recall that \cite{Che},
\cite{Scheu}, $V$ being a $\mathbb{Z}_{2}$-graded $A$-module means
that first of all $V$ is a $\mathbb{Z}_{2}$-graded vector space:
$V = V_{0} \oplus V_{1}$. Furthermore the $A$-action is such that
the even elements of $A$ act as homogeneous linear maps of degree
zero: $a \cdot V_{0} \subseteq V_{0}$ and $a \cdot V_{1} \subseteq
V_{1}$ for every $a \in A_{0}$ (even linear maps) and the odd
elements of $A$ act as homogeneous linear maps of degree one: $a
\cdot V_{0} \subseteq V_{1}$ and $a \cdot V_{1} \subseteq V_{0}$
for every $a \in A_{1}$ (odd linear maps).  \\
It is again the abstraction in the representation theory of the
quasitriangular Hopf algebras that provides us with an equivalent
description of the above mentioned ideas: We can equivalently say
that $V$ is a $\mathbb{CZ}_{2}$-module and at the same time an
$A$-module, such that the $A$ -action $\phi_{V} : A \otimes V
\rightarrow V$ is also a $\mathbb{CZ}_{2}$-module map. This is
equivalent to saying that $V$ is a (left) $A$-module in the
braided monoidal category of $\mathbb{CZ}_{2}$-modules
${}_{\mathbb{CZ}_{2}}\mathcal{M}$, or a braided $A$-module, where
the braiding is given by equation \eqref{symmbraid}. \\
The braiding \eqref{symmbraid}, provides us with an alternative
way of forming tensor products of braided modules (tensor product
of $\mathbb{Z}_{2}$-graded modules) through the same ``mechanism"
which led us to the braided tensor product algebra
\eqref{braidedtenspr}: If $A$, $B$ are algebras in
${}_{\mathbb{CZ}_{2}}\mathcal{M}$, $\ V$ is a (left) $A$-module in
${}_{\mathbb{CZ}_{2}}\mathcal{M}$, $\ W$ is a (left) $B$-module in
${}_{\mathbb{CZ}_{2}}\mathcal{M}$, via the actions:
$$
\begin{array}{ccccc}
  \phi_{V} : A \otimes V
\rightarrow V &   &   &   & \phi_{W} : B \otimes W
\rightarrow W \\
\end{array}
$$
respectively, then the vector space $V \otimes W$ becomes a (left)
$A \underline{\otimes} B$-module in
${}_{\mathbb{CZ}_{2}}\mathcal{M}$ via the action:
\begin{equation} \label{braidtensprodmod1}
\phi_{V \otimes W} = (\phi_{V} \otimes \phi_{W}) \circ (Id \otimes
\Psi_{B,V} \otimes Id) : A \underline{\otimes} B \otimes V \otimes
W \rightarrow V \otimes W
\end{equation}
given equivalently by:
$$
(a \otimes b) \cdot (v \otimes w) = (-1)^{|b||v|} a \cdot v
\otimes b \cdot w
$$
for $a, b, v, w \ $ homogeneous in $A, B, V, W$ respectively. See
also the discussion in \cite{Scheu}. Once more, we underline the
fact that in \eqref{braidtensprodmod1}, both the superalgebra $B$
and the superspace $V$ are considered as $\mathbb{CZ}_{2}$ modules
and as such: $B \otimes V \cong V \otimes B$ through $b \otimes v
\mapsto (-1)^{|b||v|}v \otimes b$, which is the
$\mathbb{Z}_{2}$-module isomorphism ``reflected" by the braiding
\eqref{symmbraid} which is the same thing as the $R_{Z_{2}}$
$R$-matrix \eqref{eq:nontrivRmatrcz2}.

\section{Super-Hopf algebraic structure of Parabosons: a braided
group}\selabel{2}

The notion of $\mathbb{G}$-graded Hopf algebra, for $\mathbb{G}$ a
finite abelian group, is not a new one neither in physics nor in
mathematics. The idea appears already in the work of Milnor and
Moore \cite{MiMo}, where we actually have $\mathbb{Z}$-graded Hopf
algebras. On the other hand, universal enveloping algebras of Lie
superalgebras are widely used in physics and they are examples of
$\mathbb{Z}_{2}$-graded Hopf algebras (see for example \cite{Ko},
\cite{Scheu}). These structures are strongly resemblant of Hopf
algebras but they are not Hopf algebras at least in the ordinary
sense.

Restricting again to the simplest case where $\mathbb{G} =
\mathbb{Z}_{2}$ we briefly recall this idea: An algebra $A$ being
a $\mathbb{Z}_{2}$-graded Hopf algebra (or super-Hopf algebra)
means first of all that $A$ is a $\mathbb{Z}_{2}$-graded
associative algebra (or: superalgebra). We now consider the
braided tensor product algebra $A \underline{\otimes} A$. Then $A$
is equipped with a coproduct
\begin{equation} \label{braidedcom}
\underline{\Delta} : A \rightarrow A \underline{\otimes} A
\end{equation}
which is an superalgebra homomorphism from $A$ to the braided
tensor product algebra  $A \underline{\otimes} A$ :
$$
\underline{\Delta}(ab) = \sum (-1)^{|a_{2}||b_{1}|}a_{1}b_{1}
\otimes a_{2}b_{2} = \underline{\Delta}(a) \cdot
\underline{\Delta}(b)
$$
for any $a,b$ in $A$, with $\underline{\Delta}(a) = \sum a_{1}
\otimes a_{2}$, $\underline{\Delta}(b) = \sum b_{1} \otimes
b_{2}$, and $a_{2}$, $b_{1}$ homogeneous. We emphasize here that
this is exactly the central point of difference between the
``super" and the ``ordinary" Hopf algebraic structure: In an
ordinary Hopf algebra $H$ we should have a coproduct $\Delta : H
\rightarrow H \otimes H$ which should be an algebra homomorphism
from $H$ to the usual tensor product algebra $H \otimes H$.

 Similarly, $A$ is equipped with an antipode $\underline{S} : A
\rightarrow A$ which is not an algebra anti-homomorphism (as in
ordinary Hopf algebras) but a  superalgebra anti-homomorphism (or:
``twisted" anti-homomorphism or: braided anti-homomorphism) in the
following sense (for any homogeneous $a,b \in A$):
\begin{equation} \label{braidedanti}
\underline{S}(ab) = (-1)^{|a||b|}\underline{S}(b)\underline{S}(a)
\end{equation}
The rest of the axioms which complete the super-Hopf algebraic
structure (i.e.: coassociativity, counity property, and
compatibility with the antipode) have the same formal description
as in ordinary Hopf algebras.

 Once again, the abstraction of the representation theory of
quasitriangular Hopf algebras provides us with a language in which
the above description becomes much more compact: We simply say
that $A$ is a Hopf algebra in the braided monoidal category of
$\mathbb{CZ}_{2}$-modules ${}_{\mathbb{CZ}_{2}}\mathcal{M}$ or: a
braided group where the braiding is given in equation
\eqref{symmbraid}. What we actually mean is that $A$ is
simultaneously an algebra, a coalgebra and a
$\mathbb{CZ}_{2}$-module, while all the structure maps of $A$
(multiplication, comultiplication, unity, counity and the
antipode) are also $\mathbb{CZ}_{2}$-module maps and at the same
time the comultiplication $\underline{\Delta} : A \rightarrow A
\underline{\otimes} A$ and the counit are algebra morphisms in the
category ${}_{\mathbb{CZ}_{2}}\mathcal{M}$ (see also \cite{Maj2,
Maj3} or \cite{Mon} for a more detailed description).
\\
We proceed now to the proof of the following proposition which
establishes the super-Hopf algebraic structure of the parabosonic
algebra $P_{B}$:
\begin{proposition} \prlabel{superHopfPb}
The parabosonic algebra  equipped with the even linear maps
$\underline{\Delta}: P_{B} \rightarrow P_{B} \underline{\otimes}
P_{B} \ \ $, $\ \ \underline{S}: P_{B} \rightarrow P_{B} \ \ $, $\
\ \underline{\varepsilon}: P_{B} \rightarrow \mathbb{C} \ \ $,
determined by their values on the generators:
\begin{equation} \label{eq:HopfPB}
\begin{array}{ccccc}
  \underline{\Delta}(B_{i}^{\pm}) = 1 \otimes B_{i}^{\pm} + B_{i}^{\pm} \otimes 1 &
  & \underline{\varepsilon}(B_{i}^{\pm}) = 0  & & \underline{S}(B_{i}^{\pm}) = - B_{i}^{\pm} \\
\end{array}
\end{equation}
for $i = 1, 2, \ldots \ $, becomes a super-Hopf algebra.
\end{proposition}
\begin{proof}
Recall that by definition $P_{B} = T(V_{X}) / I_{P_{B}}$. Consider
the linear map: $$\underline{\Delta}^{T}: V_{X} \rightarrow P_{B}
\underline{\otimes} P_{B}$$ determined by it's values on the basis
elements specified by: $$\underline{\Delta}^{T}(X_{i}^{\pm}) = I
\otimes B_{i}^{\pm} + B_{i}^{\pm} \otimes I$$ By the universality
of the tensor algebra this map is uniquely extended to a
superalgebra homomorphism: $\underline{\Delta}^{T}: T(V_{X})
\rightarrow P_{B} \underline{\otimes} P_{B}$. Now we compute:
$$
\underline{\Delta}^{T}(\big[ \{ X_{i}^{\xi},  X_{j}^{\eta}\},
X_{k}^{\epsilon}  \big] -
 (\epsilon - \eta)\delta_{jk}X_{i}^{\xi} - (\epsilon -
 \xi)\delta_{ik}X_{j}^{\eta})= 0
$$
for all values of $\xi, \eta, \epsilon = \pm 1$ and $i,j=1, 2, ...
\ $. This means that $I_{P_{B}} \subseteq
ker\underline{\Delta}^{T}$, which in turn implies that
$\underline{\Delta}^{T}$ is uniquely extended as a superalgebra
homomorphism: $\underline{\Delta}: P_{B} \rightarrow P_{B}
\underline{\otimes} P_{B}$, according to the following
(commutative) diagram:
\begin{displaymath}
\xymatrix{T(V_{X}) \ar[rr]^{\underline{\Delta}} \ar[dr]_{\pi} & &
P_{B} \underline{\otimes} P_{B} \\
 & P_{B} \ar@{.>}[ur]_{\exists ! \ \underline{\Delta}} & }
\end{displaymath}
with values on the generators determined by \eqref{eq:HopfPB}.

Proceeding the same way we construct the maps $\
\underline{\varepsilon} \ $, $\ \ \underline{S} \ $, as determined
in \eqref{eq:HopfPB}. Note here that in the case of the antipode
$\underline{S}$ we need the notion of the $\mathbb{Z}_{2}$-graded
opposite algebra (or: opposite superalgera) $P_{B}^{op}$, which is
a superalgebra defined as follows: $P_{B}^{op}$ has the same
underlying super vector space as $P_{B}$, but the multiplication
is now defined as: $a \cdot b = (-1)^{|a||b|}ba$, for all $a,b \in
P_{B}$. (In the right hand side, the product is of course the
product of $P_{B}$). We start by defining a linear map
$$\underline{S}^{T}: V_{X} \rightarrow P_{B}^{op}$$ determined by:
$$\underline{S}^{T}(X_{i}^{\pm}) = -B_{i}^{\pm}$$ This map is (uniquely)
extended to a superalgebra homomorphism: $\underline{S}^{T}:
T(V_{X}) \rightarrow P_{B}^{op}$. We proceed by showing that:
\begin{equation} \label{eq:antipodeinsuperparab}
\underline{S}^{T}(\big[ \{ X_{i}^{\xi},  X_{j}^{\eta}\},
X_{k}^{\epsilon}  \big] -
 (\epsilon - \eta)\delta_{jk}X_{i}^{\xi} - (\epsilon -
 \xi)\delta_{ik}X_{j}^{\eta}) = 0
\end{equation}
for all values of $\xi, \eta, \epsilon = \pm 1$ and $i,j=1, 2, ...
\ $, or in other words: $I_{P_{B}} \subseteq ker\underline{S}^{T}$
which in turn implies that $\underline{S}^{T}$ is uniquely
extended to a superalgebra homomorphism $\underline{S}: P_{B}
\rightarrow P_{B}^{op}$, thus to a superalgebra anti-homomorphism:
$\underline{S}: P_{B} \rightarrow P_{B}$ with values on the
generators determined by \eqref{eq:HopfPB}.

Now it is sufficient to verify the rest of the super-Hopf algebra
axioms (coassociativity, counity and the compatibility condition
for the antipode) on the generators of $P_{B}$. This can be done
with straigthforward computations.
\end{proof}
Let us note here, that the above proposition generalizes a result
which -in the case of finite degrees of freedom- is a direct
consequence of the work in \cite{Pal}.

Finally, before closing this paragraph we should mention an
important consequence of the super-Hopf algebraic structure at the
level of representations. If $A$ is a super-Hopf algebra, the
existence of the comultiplication \eqref{braidedcom} permits us to
construct tensor product of braided representations (tensor
product of $\mathbb{Z}_{2}$-graded representations), through the
mechanism provided by the braiding and developed at the end of
\seref{11}: If $A$ is a Hopf algebra in
${}_{\mathbb{CZ}_{2}}\mathcal{M}$ and $V$, $\ W$, are (left)
$A$-modules in ${}_{\mathbb{CZ}_{2}}\mathcal{M}$, via the actions:
$$
\begin{array}{ccccc}
  \phi_{V} : A \otimes V
\rightarrow V &   &   &   & \phi_{W} : A \otimes W
\rightarrow W \\
\end{array}
$$
then the vector space $V \otimes W$ becomes a (left) $A$-module in
${}_{\mathbb{CZ}_{2}}\mathcal{M}$, via the action:
\begin{equation} \label{braidtensprodmod}
\phi_{V \otimes W} = (\phi_{V} \otimes \phi_{W}) \circ (Id \otimes
\Psi_{A,V} \otimes Id) \circ (\underline{\Delta} \otimes Id
\otimes Id) : A \otimes V \otimes W \rightarrow V \otimes W
\end{equation}
given equivalently by:
$$
a \cdot (v \otimes w) = \sum (-1)^{|a_{2}||v|}a_{1} \cdot v
\otimes a_{2} \cdot w
$$
for $a, v, w \ $ homogeneous in $A, V, W$ respectively, and
$\underline{\Delta}(a) = \sum a_{1} \otimes a_{2}$ for any $a \in
A$.

The preceding construction, can be generalized for longer tensor
products of braided representations: If $(V_{k})_{k \in I}$ ($I$:
some finite set) is a collection of (left) $A$-modules in
${}_{\mathbb{CZ}_{2}}\mathcal{M}$ for $A$ some super-Hopf algebra,
then the vector space $\otimes_{i \in I}V_{i}$ readily becomes an
$A \underline{\otimes} ... \underline{\otimes} A$-module
($I$-copies of $A$) according to the discussion at the end of
\seref{1}. The only thing we have to keep in mind is that each
time $A$ and $V_{i}$ have to interchange their order in the tensor
product $A \underline{\otimes} ... \underline{\otimes} A
\otimes_{i \in I}V_{i}$, in order for the action to be computed,
the suitable sign factor due to the braiding \eqref{symmbraid} has
to be inserted. (We underline the fact that in such permutations
between a copy of $A$ and $V_{i}$ both of $A$ and $V_{i}$ are
considered as $\mathbb{CZ}_{2}$-modules). Repeated use of the
comultiplication of the super-Hopf algebra $A$ may then be used to
turn this $A \underline{\otimes} ... \underline{\otimes} A$-module
to an $A$-module. We will come back in \seref{brainterpGrAns} to
discuss some applications and a more detailed formulation of these
ideas in the case of the tensor product representations of the
parabosonic algebra $P_{B}$.

\section{Bosonic and Parabosonic modules} \selabel{Gr}

 \subsection{Fock representation for Bosons} \selabel{bose}

We start our discussion by considering for simplicity, the case of
one degree of freedom. We thus consider the Weyl algebra (or :
CCR) \eqref{CCR} determined by
\begin{equation} \label{CCR1degree}
[q, p] = i \hbar I
\end{equation}
The generalization to the case of the finite degrees of freedom
(finite number of bosons) is straightforward as we shall see.

Since the early days of quantum theory, the analysis of the
representations of such an algebra underlies much of the
foundations of quantum mechanics. It was in the mid-twenties that
Heisenberg himself \cite{Heis} and Schr\"{o}edinger \cite{Schroed}
stated explicitly the first representation of such an algebra:
Guided by physical requirements, they considered representations
of \eqref{CCR1degree} on a complex, separable, infinite
dimensional Hilbert space, where $p$ and $q$ act as unbounded
\footnote{Generally speaking, an old result \cite{Wint} (but see
also in \cite{Put}) states that when considering representations
of the CCR \eqref{CCR1degree} in a Hilbert space, the generators
$p$ and $q$ cannot be both represented by bounded operators: at
least one of them has to be unbounded. This fact automatically
raises delicate questions about the domains of such operators and
at the same time indicates the importance of the use of unbounded
operators in the mathematical formulation of quantum mechanics. An
analogous result for the parabosonic algebra has been obtained in
\cite{Rob}.}, self-adjoint operators. Of course such a Hilbert
space is unique and we can regard as -isomorphic- copies of it,
either the space $l_{2}$ of all square integrable complex
sequences i.e. sequences $(x_{1}, x_{2}, ... ,x_{n}, ... )$ such
that $\sum_{i=1}^{\infty}|x_{1}|^{2} < \infty$, or the space
$L_{2}(- \infty, \infty)$ of all square integrable complex valued
functions i.e. functions $f(x)$ such that $\int_{-
\infty}^{\infty}|f(x)|^{2}dx < \infty$, where $x$ is a real
variable. \\
Heisenberg made the first choice and using as ``basis" of $l_{2}$
the total orthonormal set consisting of the elements $e_{i} =
(0,...,0,1,0,...)$ with $1$ on the $i$-th place and zero
everywhere else ($i = 1, 2, ...$) represented the generators $q$,
$ p$, in the form of infinite dimensional matrices. That was the
origin of the so-called matrix mechanics formalism. \\
On the other hand Schr\"{o}edinger made use of the second option
and represented the generators $q$, $p$, as the multiplication and
the differentiation operators $x$ and $-i \frac{d}{dx}$ in
$L_{2}(- \infty, \infty)$ respectively, resulting thus with what
is now known as the wave mechanics formalism. \\
 These two
representations are (unitarily) equivalent: If we choose a total
orthonormal set in $L_{2}(- \infty, \infty)$ consisting of the
functions $e_{n}(x) = e^{-x^{2}/2} H_{n}(x) \big/
\pi^{1/4}\sqrt{2^{n}n!}$ where $H_{n}(x)$ are the Hermite
polynomials ($n = 0, 1, 2, ...$) we can straightforwardly check
that the matrix forms of the operators $x$ and $-i \frac{d}{dx}$
in $L_{2}(- \infty, \infty)$ coincide with the infinite
dimensional matrices stated by Heisenberg. The corresponding
modules are thus
isomorphic. \\

We must point out here that the Heisenberg - Schr\"{o}edinger
representation described above is not a representation in the
strict algebraic sense. The multiplication and differentiation
operators are obviously not defined on the entire of $L_{2}(-
\infty, \infty)$. Hopefully the situation is easily cured: Neither
the sequences $e_{i}$, $i = 1, 2, ...$ nor the functions
$e_{n}(x)$, $n = 0, 1, 2, ...$ constitute bases -in the sense of
algebra- of the spaces $l_{2}$ and $L_{2}(- \infty, \infty)$
respectively. In other words they are not Hammel bases but simply
total orthonormal sequences. This means that their span is not the
entire space but only a dense subspace of it. This subspace is no
longer a Hilbert space, since it is not complete. It is merely a
euclidean (or: pre-Hilbert) space. Restricting our attention to
this space, the Heisenberg - Schr\"{o}edinger representation is
actually well-defined (in the strict algebraic sense) and it is a
well-known fact for physicists that it is an irreducible
representation (i.e: a simple module). If we use the basis given
in \eqref{creation-destruction op} instead of \eqref{CCR} we have
-for the case of a single degree of freedom- the CCR in the form
of the single boson algebra:
\begin{equation} \label{CCRbose1deg}
[b^{-}, b^{+}] = I
\end{equation}
Following a widespread notation of the physical literature, we
shall denote the $e_{i}$, $i = 1, 2, ...$ or the $e_{n}(x)$, $n =
0, 1, 2, ...$ total orthonormal sets of $l_{2}$ or $L_{2}(-
\infty, \infty)$ respectively by $|n>$ for $n = 0, 1, 2, ...$.
Using this formalism it is straightforward to check that the
Heisenberg - Schr\"{o}edinger representation can be equivalently
described in terms of the following action:
\begin{equation} \label{Fock1particle}
\begin{array}{ccccc}
  b^{+} |n> = \sqrt{n+1} \ |n+1> &  &  &  & b^{-} |n> = \sqrt{n} \ |n-1> \\
\end{array}
\end{equation}
where the elements $b^{+}$, $ \ b^{-}$, act as Hilbert-adjoint
operators, i.e.:
$$
(b^{-})^{\dagger} = b^{+}
$$
The carrier space is the pre-Hilbert space generated by the set
$\{ |n> / n = 0, 1, 2, ... \}$. This is a dense subspace of the
complex separable infinite dimensional Hilbert space $H$. We can
easily check using the above described formalism that:
\begin{equation} \label{Fock1particlegenerat}
\begin{array}{ccccc}
 |n> = \frac{1}{\sqrt{n!}}(b^{+})^{n} |0>  &  &  &  & b^{-} |0> = 0 \\
\end{array}
\end{equation}
which enables us to conclude that the Heisenberg -
Schr\"{o}edinger representation is a cyclic module generated by
any of it's elements, which in turn implies that it is a simple
module (or: an irreducible representation). The vector $|0>$ is
called the ``ground state" or the ``vacuum" of the system and it's
corresponding wave-function (under the above mentioned
isomorphism) is: $\ \pi^{-1/4} e^{-x^{2}/2}$. \\

The generalization of the above mentioned representation for the
general case of the bosonic algebra is a well known fact for
physicists:

 In the case of a finite number of bosons we have the
algebra $B^{(m)}$ (bosonic algebra in $2m$ generators or $m-$boson
algebra) described in terms of generators and relations by
\eqref{CCRbose} for $\ i,j = 1, 2, ... , m$. We construct a
representation which is uniquely determined (up to unitary
equivalence) by the demand for the existence of a unique vacuum
vector belonging in a Hilbert space and annihilated by all
$b_{i}^{-}$'s, together with the Hilbert adjointness condition:
\begin{equation} \label{Fockfinite1}
\begin{array}{ccccc}
 b_{i}^{-} |0> = 0 &   &   &   & (b_{i}^{-})^{\dagger} = b_{i}^{+}  \\
\end{array}
\end{equation}
for every value $i = 1, 2, ..., m$. The carrier space is generated
by elements of the form:
\begin{equation} \label{Fockm}
|k_{1}, k_{2}, \cdots, k_{i}, \cdots, k_{m}> =
\frac{(b_{1}^{+})^{k_{1}}(b_{2}^{+})^{k_{2}} \cdots
(b_{i}^{+})^{k_{i}} \cdots (b_{m}^{+})^{k_{m}}
}{\sqrt{k_{1}!k_{2}! \ldots k_{i}! \ldots k_{m}!}} |0>
\end{equation}
where $i = 1, 2, ..., m$, the $k_{i}$'s are non-negative integers.
The corresponding actions of the elements of $B^{(m)}$ can be
calculated to be:
\begin{equation} \label{Fockactm}
\begin{array}{c}
  b_{i}^{+} |k_{1}, k_{2}, \cdots, k_{i}, \cdots, k_{m}> =
\sqrt{k_{i}+1} \ |k_{1}, k_{2}, \cdots, k_{i}+1, \cdots, k_{m}> \\
   \\
  b_{i}^{-} |k_{1}, k_{2}, \cdots, k_{i}, \cdots, k_{m}> =
\sqrt{k_{i}} \ |k_{1}, k_{2}, \cdots, k_{i}-1, \cdots, k_{m}> \\
\end{array}
\end{equation}
This is again a cyclic module generated by any of it's elements,
thus a simple module. This representation is usually called the
Fock representation or the Fock-Cook representation. We are going
to denote it by ${}_{B^{(m)}}F$. Certain subspaces of this vector
space contain the physically realizable states i.e. vectors which
are symmetric under the exchange of particles leading thus to the
Bose-Einstein statistics.  The Fock representation ${}_{B^{(m)}}F$
can be constructed as the tensor product of $m-$copies of the
$1-$particle Heisenberg - Schr\"{o}edinger representation
${}_{B^{(1)}}F$ described by \eqref{Fock1particle},
\eqref{Fock1particlegenerat}, where we do not take into account
the $\mathbb{Z}_{2}-$grading when constructing the tensor product.
We will not pursue these subjects further here but instead we
refer to the classic works \cite{FockCook, Seg}. \\
We mention that the Fock representation ${}_{B^{(m)}}F$ is a
$\mathbb{Z}_{2}-$graded $B^{(m)}-$module (or equivalently a
$B^{(m)}-$module in the braided monoidal category of the
$\mathbb{CZ}_{2}$-modules ${}_{\mathbb{CZ}_{2}}\mathcal{M}$). This
can be easily seen since the vector space spanned by elements of
the form \eqref{Fockm} is a $\mathbb{Z}_{2}-$graded vector space.
It's even subspace is spanned by elements with
$\sum_{i=1}^{m}k_{i} = even$ and it's odd subspace is spanned by
elements with $\sum_{i=1}^{m}k_{i} = odd$. The action given in
\eqref{Fockactm} describes how the (odd) generators of $B^{(m)}$
act on the above described subspaces.  \\

In the case of infinite degrees of freedom, we have the general
case of the bosonic algebra $B \ $ described in terms of
generators and relations by \eqref{CCRbose} for $\ i,j = 1, 2, ...
\ $ (see also \deref{parabosonsbosons}). In this case the Fock
space is generated by elements of the form:
\begin{equation} \label{Fockinfinitegenerated}
|k_{1}, k_{2}, \cdots, k_{i}, \cdots> =
\frac{(b_{1}^{+})^{k_{1}}(b_{2}^{+})^{k_{2}} \cdots
(b_{i}^{+})^{k_{i}} \cdots}{\sqrt{k_{1}!k_{2}! \ldots k_{i}!
\ldots}} |0>
\end{equation}
and the Fock representation ${}_{B}F$ is uniquely determined (up
to unitary equivalence) by the demand for the existence of a
unique vacuum together with the Hilbert adjointness condition:
\begin{equation} \label{Fockinfinite1}
\begin{array}{ccccc}
 b_{i}^{-} |0> = 0 &   &   &   & (b_{i}^{-})^{\dagger} = b_{i}^{+}  \\
\end{array}
\end{equation}
which in turn imply:
\begin{equation} \label{Fockinfinite}
\begin{array}{c}
 b_{i}^{+} |k_{1}, k_{2}, \cdots, k_{i}, \cdots> =
\sqrt{k_{i}+1} \ |k_{1}, k_{2}, \cdots, k_{i}+1, \cdots>    \\
    \\
   b_{i}^{-} |k_{1}, k_{2}, \cdots, k_{i}, \cdots> =
\sqrt{k_{i}} \ |k_{1}, k_{2}, \cdots, k_{i}-1, \cdots>  \\
\end{array}
\end{equation}
where $i = 1, 2, ...$ and the $k_{i}$'s are non-negative integers.
\\
 The monomials $(b_{1}^{+})^{k_{1}}(b_{2}^{+})^{k_{2}} \cdots
(b_{i}^{+})^{k_{i}} \cdots \ $ stated in the right hand side of
\eqref{Fockinfinitegenerated} contain of course a finite number of
factors only. In other words for any vector of the form
\eqref{Fockinfinitegenerated} we have: $\sum_{j=1}^{\infty}k_{j} <
\infty$. We have an irreducible, $\mathbb{Z}_{2}-$graded,
$B-$module which we will denote by ${}_{B}F$. This is actually the
starting point for the study of the second quantization problem
(see: \cite{FockCook, Seg}).

\subsection{The Green's Ansatz and Fock-like representations for
Parabosons} \selabel{GrAns}

When Green first introduced the parabosonic algebra in infinite
degrees of freedom $P_{B}$ in \cite{Green}, one of his first tasks
was the investigation of it's representations and more
specifically the construction of representations which might serve
as a generalization of the bosonic Fock representation previously
described. For this purpose, he introduced a device which has been
known since then in the bibliography of mathematical physics as
the ``Green's ansatz".

The Green's ansatz $(G(p))_{p=1, 2, ...}$ is an infinite family of
associative algebras labelled by a positive integer $p$. For each
specific value of $p$ we have a different algebra denoted by
$G(p)$ described by means of generators and relations: It is the
associative algebra generated by the elements $b_{i}^{\alpha}$, $\
b_{i}^{\alpha \ +}$ and the unity $I$, for all values $i = 1, 2,
... \ $ and $\alpha = 1, 2, ..., p$. For a specific value of
$\alpha$ the generators satisfy the usual CCR :
\begin{equation} \label{GrAns1}
\begin{array}{ccccc}
  [b_{i}^{\alpha \ -}, b_{j}^{\alpha \ +}] = \delta_{ij}I &   &   &   &
  [b_{i}^{\alpha \ -}, b_{j}^{\alpha \ -}] = [b_{i}^{\alpha \ +}, b_{j}^{\alpha \ +}] = 0 \\
\end{array}
\end{equation}
while for values $\alpha \neq \beta$ we have the ``anomalous"
(anticommutation) relations:
\begin{equation} \label{GrAns2}
  \{ b_{i}^{\alpha \ -}, b_{j}^{\beta \ +} \} =
  \{ b_{i}^{\alpha \ -}, b_{j}^{\beta \ -} \} = \{ b_{i}^{\alpha \ +}, b_{j}^{\beta \ +} \} = 0
\end{equation}
The above relations hold for all values of $i, j = 1, 2, ... \ $
and $\alpha = 1, 2, ..., p$ and completely describe the algebra
$G(p)$. In a terminology more familiar to physicists, we can say
that $G(p)$ describes $p$ anticommuting bosonic fields. We note
here, that for $p=1$, the corresponding algebra of Green's ansatz
is just the familiar bosonic algebra: $G(1) = B$.

Green showed that for every specific algebra $G(p)$ we pick, among
the ones constituting the Green's ansatz, i.e. for every specific
choice we make for the value of the positive integer $p$, we have
an associative algebra homomorphism $\pi_{p} : P_{B} \rightarrow
G(p)$ stated explicitly by:
\begin{equation} \label{GrAnsRepr}
\pi_{p}(B_{i}^{\pm}) = \sum_{\alpha = 1}^{p} b_{i}^{\alpha \ \pm}
\end{equation}
We remark here, that for $p=1$, we have: $\pi_{1} = \phi : P_{B}
\rightarrow B = G(1)$. This is exactly the replacement map of
\prref{parabosonstobosons}, which is nothing else but the algebra
projection of $P_{B}$ onto $B$.

 Green's idea was to make use of
the above mentioned homomorphisms, in order to construct
representations of the parabosonic algebra $P_{B}$ initiating from
representations of $G(p)$. However his idea was not fully
exploited until a few years later by Greenberg and Messiah, in
\cite{GreeMe}, where they virtually classified all the Fock-like
representations of $P_{B}$. Their method was -shortly- as follows:
For a specific value of $p$ they considered a representation of
the $G(p)$ in (a subspace of) a Hilbert space, determined by the
demand for the existence of a unique ``vacuum" vector $|0>$ :
\begin{equation} \label{GrAnsvac}
 b_{i}^{\alpha \ -} |0> = 0
\end{equation}
under the Hilbert-adjointness condition:
\begin{equation} \label{GrAnsvac1}
(b_{i}^{\alpha \ -})^{\dagger} = b_{i}^{\alpha \ +}
\end{equation}
for all values $i = 1, 2, ... \ $ and $\alpha = 1, 2, ..., p$.
They showed that the cyclic $G(p)$-module generated by $|0>$, is
an irreducible representation of the $G(p)$ algebra and is
specified -up to unitary equivalence- by the conditions
\eqref{GrAnsvac}, \eqref{GrAnsvac1}. This representation has as
carrier space the vector space spanned by vectors of the form:
$\mathcal{P}(b_{i}^{\alpha \ \pm})|0>$, where $\mathcal{P}$ is an
arbitrary polynomial of the generators $b_{i}^{\alpha \ \pm}$ for
all values $i = 1, 2, ... \ $ and $\alpha = 1, 2, ..., p$.
Utilizing the algebra homomorphism \eqref{GrAnsRepr} they turned
the above constructed vector space into a $P_{B}$-module through:
$B_{i}^{\pm}|..> = \pi_{p}(B_{i}^{\pm})|..>$ for all values $i =
1, 2, ... \ $. They showed that this $P_{B}$-module is reducible,
and that it contains an irreducible $P_{B}$-submodule which is
cyclic and generated (as a submodule) by $|0>$. This irreducible
$P_{B}$-submodule is generated (as a vector space) by all the
elements of the form: $\mathcal{P}(B_{i}^{\pm})|0>$ for
$\mathcal{P}$ an arbitrary polynomial of the generators.
Furthermore they showed that in this (irreducible)
$P_{B}$-submodule we have the relations:
\begin{equation} \label{FocklikePb}
\begin{array}{ccccc}
  B_{i}^{-} |0> = 0 &  & B_{i}^{-}B_{j}^{+}|0> = p \delta_{ij} |0> &  & (B_{i}^{-})^{\dagger} = B_{i}^{+} \\
\end{array}
\end{equation}
for all values $i, j = 1, 2, ... \ $, which specify it up to
unitary equivalence.   \\

 Apart from utilizing Green's ansatz, Greenberg and
Messiah also studied the representations of the parabosonic
algebra in more abstract terms. Inspired by the case of the
bosonic algebra (CCR), they studied Fock-like representations of
the parabosonic algebra i.e. representations of $P_{B}$ in a
(suitably chosen subspace of a) Hilbert space possessing a unique
``vacuum" which means a unique vector satisfying: $B_{i}^{-} |0> =
0$ under the condition that $B_{i}^{-}$, $B_{i}^{+}$, act as
Hilbert adjoint operators: $(B_{i}^{-})^{\dagger} = B_{i}^{+}$,
for all values $i = 1, 2, ... \ $. They found that unlike the
boson case where such requirements specify a unique up to
isomorphism irreducible representation, the Fock-representation
${}_{B}F$, in the paraboson case these requirements specify -up to
isomorphism- an infinite collection of irreducible representations
labelled by a positive integer $p$. We are going to denote these
Fock-like representations by ${}_{P_{B}}F(p)$. Their central
results are summarized in the following theorem:
\begin{theorem} \thlabel{FocklikereprPb}
Any representation of the parabosonic algebra $P_{B}$ in a Hilbert
space, possessing a unique ``vacuum" vector satisfying: $B_{i}^{-}
|0> = 0$ for all values $i = 1, 2, ... \ $ and for which:
$(B_{i}^{-})^{\dagger} = B_{i}^{+} \ $ $\forall$ $ \ i = 1, 2,
...$, has also the following properties:
\begin{itemize}
\item[-] It also satisfies: $B_{i}^{-}B_{j}^{+}|0> = p \delta_{ij}
|0>$ for all values $i, j = 1, 2, ... \ $, where $p$ is an
arbitrary positive integer ($p = 1, 2, ... \ $).

\item[-] For a fixed value of $p$, it is an irreducible
representation (i.e.: the corresponding module is simple).

\item[-] For a fixed value of $p$ this irreducible representation
is unique up to unitary equivalence (and thus is isomorphic with
the $p$-submodule contained in the reducible module produced by
Green's ansatz previously described).

\item[-] The carrier space is spanned by vectors of the form:
$\mathcal{P}(B_{i}^{+})|0>$ where $\mathcal{P}$ an arbitrary
monomial of the generators $B_{i}^{+}$ for all values $i = 1, 2,
... \ $, and it is an infinite dimensional vector space.
\end{itemize}
\end{theorem}
Thus: the simple $P_{B}$-modules satisfying the ``vacuum" state
(or: ``ground" state) condition $B_{i}^{-}|0> = 0$, are labelled
by a positive integer $p$. In the physics literature this integer
is called ``the order of the parastatistics". These are the
Fock-like representations ${}_{P_{B}}F(p)$ of the parabosonic
algebra $P_{B}$ generalizing the bosonic Fock representation
${}_{B}F$ described in the previous paragraph. For a detailed
proof of the theorem without the use of Green's ansatz see also
\cite{OhKa} chapter $4$, while the proof via the Green's ansatz
can be also found in chapter $5$ of the same reference.  \\
For $p = 1$ we get the first Fock-like representation
${}_{P_{B}}F(1)$ of the parabosonic algebra, which we will denote
by ${}_{P_{B}}F$ from now on. We can further check that for $p=1$
we have $G(1) = B$ and the Green homomorphism $\pi_{1}$ of
\eqref{GrAnsRepr} coincides with the replacement map
$\phi(B_{i}^{\pm}) = b_{i}^{\pm}$ of \prref{parabosonstobosons}.
This implies, that the representation ${}_{P_{B}}F$ coincides with
the representation of $P_{B}$ constructed via the bosonic Fock
representation ${}_{B}F$ and \coref{Parabreprfrombosrepr}.

 Let us
finally emphasize that Greenberg's work allows one to specify a
much more general class of representations than the Fock-like
representations described in \thref{FocklikereprPb}: If the
requirement of the positivity of the inner product is dropped,
i.e. if we consider more general spaces than Hilbert spaces (for
example in Krein spaces), then only the first and the last
conclusions of the above theorem hold, but with p now being an
arbitrary positive number. We won't deal with such representations
here.     \\

 We should mention here that the carrier space for
every Fock-like representation of the parabosonic algebra (this
means: for every specific value of the positive integer $p$) has a
much more complicated structure than the corresponding carrier
space of the bosonic Fock representation: In the latter case, the
fact that the generators $b_{i}^{+}$ commute with each other for
every value of $i = 1, 2, ... \ $ means that an arbitrary monomial
of the generators $b_{i}^{+}$ is of the form
$(b_{1}^{+})^{k_{1}}(b_{2}^{+})^{k_{2}} \cdots (b_{i}^{+})^{k_{i}}
\cdots$ (with a finite number of factors only, as has already been
stated). This in turn implies that the positive integers $k_{1},
k_{2}, ..., k_{i}, ...$ ($k_{i} \neq 0$ for a finite number of
values of $i$ only) uniquely determine the vectors spanning the
bosonic Fock space. Consequently, the notation used in equations
\eqref{Fockm} and \eqref{Fockinfinitegenerated} is well defined.
It is obvious however that this is not the case for the
parabosonic Fock-like representations: The more complicated nature
of the relations in the parabosonic algebra imply that in general,
even vectors of the form: $B_{i}^{+}B_{j}^{+}|0>$ and
$B_{j}^{+}B_{i}^{+}|0>$ for $i \neq j$ may be linearly
independent. From the physicist's point of view, one may say that
permutations of the $b_{i}^{+}, b_{j}^{+}, ... \ $ correspond to
permutations of particles (bosons) while in general permutations
of the $B_{i}^{+}, B_{j}^{+}, ... \ $ do not correspond to
permutations of particles (parabosons). Since the construction of
the Fock space is ultimately connected with the statistics of the
particles, this immediately draws the line between ordinary
Bose-Einstein statistics and parastatistics, but this is an
entirely different subject and we will not discuss it here.

  A consequence of the preceding conversation, is that the
computation of the explicit form of the action of the generators
$B_{i}^{+}, B_{j}^{-}$, on vectors of the form
$\mathcal{P}(B_{i}^{+})|0>$ which span the carrier space for the
parabosonic Fock-like representation, is not a trivial matter at
all. Up to our knowledge, it has been done only for the case of a
single paraboson (i.e.: the parabosonic algebra $P_{B}^{(1)}$
generated by the elements $B^{+}, B^{-}$ subject to the relations
\eqref{eq:pbdef} for $i = 1$). Such computations can be found in
\cite{OhKa}. They lead to the following relations:
\begin{equation}
\begin{array}{cc}
  B^{+} |2n> = \sqrt{2n+p} \ |2n+1> &  B^{+} |2n+1> = \sqrt{2n+2} \ |2n+2>  \\
    \\
  B^{-} |2n> = \sqrt{2n} \ |2n-1>  &  B^{-} |2n+1> = \sqrt{2n+p} \ |2n>  \\
\end{array}
\end{equation}
for $n = 0, 1, 2, ... \ $ and $|n> \sim (B^{+})^{n} |0>$, which
explicitly describe the action of the single paraboson algebra on
it's Fock-like representation determined by the specific value of
$p$.

Finally it is worth noticing before closing this paragraph, that
the Fock-like representations ${}_{P_{B}}F(p)$ of the parabosonic
algebra $P_{B}$ are $\mathbb{Z}_{2}-$graded $P_{B}-$modules (or
equivalently $P_{B}-$modules in the braided monoidal category of
the $\mathbb{CZ}_{2}$-modules ${}_{\mathbb{CZ}_{2}}\mathcal{M}$).

\section{Braided bosons: The Green's ansatz revisited} \selabel{brainterpGrAns}

We turn in this paragraph to an application of the braided tensor
product ``technology" in the theory of the representations of the
parabosonic algebra $P_{B}$. We prove the following proposition
which gives a ``braided" interpretation of the Green's ansatz
device:
\begin{proposition} \prlabel{braidGrAns}
For every specific choice of the positive integer $p$, the
corresponding algebra $G(p)$ in Green's ansatz, which is described
in terms of generators and relations in equations \eqref{GrAns1},
\eqref{GrAns2}, is isomorphic to the braided tensor product
algebra, of $p$ copies of the bosonic algebra $B$. We have the
superalgebra isomorphism:
\begin{equation} \label{isomGrAnsbraidtenspro}
G(p) \cong B \underline{\otimes} \ldots \underline{\otimes} B
\end{equation}
where we have $p$-copies of the bosonic algebra $B$ in the right
hand side of the above relation.
\end{proposition}
\begin{proof}
Let us use -for the elements of the braided tensor product algebra
$B \underline{\otimes} \ldots \underline{\otimes} B \ $- the
notation:
$$
I \otimes ... \otimes I \otimes b_{i}^{\pm} \otimes I \otimes I
\otimes ... \otimes I \equiv b_{i}^{(r) \pm}
$$
where in the left hand side of the above there are $p$-factors,
$b_{i}^{\pm}$ is placed on the $r$-th place and the identity
element $I$ of the bosonic algebra $B$ is everywhere else, for all
values of $i = 1, 2, ... \ $ and $r = 1, 2, ... , p \ $. The
elements $b_{i}^{(r) \pm}$ together with the identity $I \otimes I
\otimes ... \otimes I$ constitute a set of generators of the
braided tensor product algebra $B \underline{\otimes} \ldots
\underline{\otimes} B$ ($p$-copies of $B$). We proceed now to
compute the multiplication in the braided tensor product algebra.
We have for all values of $i , j = 1, 2, ... \ $ and $r = 1, 2,
... , p \ $
$$
\begin{array}{c}
  b_{i}^{(r) -}b_{j}^{(r) +} = I \otimes ... \otimes I \otimes
b_{i}^{-}b_{j}^{+} \otimes I \otimes I \otimes ... \otimes I \\
    \\
 b_{j}^{(r) +}b_{i}^{(r) -} = I \otimes ... \otimes I \otimes
b_{j}^{+}b_{i}^{-} \otimes I \otimes I \otimes ... \otimes I  \\
\end{array}
$$
where on the right hand side of the above equations
$b_{i}^{-}b_{j}^{+}$ and $b_{j}^{+}b_{i}^{-}$ respectively, lie on
the $r$-th place. Subtracting the above equations by parts finally
gives:
\begin{equation} \label{BrGrAns11}
\begin{array}{c}
  [b_{i}^{(r) \ -}, b_{j}^{(r) \ +}] = I \otimes ... \otimes I
\otimes [b_{i}^{-}, b_{j}^{+}] \otimes I \otimes I \otimes ...
\otimes I = \\
    \\
 = I \otimes ... \otimes I \otimes  \delta_{ij}I \otimes
I \otimes I \otimes ... \otimes I = \delta_{ij}(I \otimes ...
 \otimes I) \\
\end{array}
\end{equation}
In the right hand side of the above $I \otimes ... \otimes I$ is
the identity element of the braided tensor product algebra $B
\underline{\otimes} \ldots \underline{\otimes} B$.  In the same
way we also have:
\begin{equation} \label{BrGrAns12}
[b_{i}^{(r) \ +}, b_{j}^{(r) \ +}] = [b_{i}^{(r) \ -}, b_{j}^{(r)
\ -}] = 0
\end{equation}
in $B \underline{\otimes} \ldots \underline{\otimes} B$. \\
Proceeding in the same way, we compute for all values of $i , j =
1, 2, ... \ $ and $r , s = 1, 2, ... , p \ $ with $r < s$ :
$$
\begin{array}{c}
  b_{i}^{(r) -}b_{j}^{(s) +} = I \otimes ... \otimes I \otimes
b_{i}^{-} \otimes I \otimes ... \otimes I \otimes b_{j}^{+}
\otimes I \otimes ... \otimes I \\
    \\
 b_{j}^{(s) +}b_{i}^{(r) -} = - I \otimes ... \otimes I \otimes
b_{i}^{-} \otimes I \otimes ... \otimes I \otimes b_{j}^{+}
\otimes I \otimes ... \otimes I   \\
\end{array}
$$
On the right hand side of the above equations, the elements
$b_{i}^{-}$ and $b_{j}^{+}$ lie on the $r$-th and the $s$-th
places respectively. The minus sign in the second of the above
equations is inserted due to the braiding \eqref{symmbraid} of
\seref{1}: The reason is that since $r < s$, when computing the
product $b_{j}^{(s) +}b_{i}^{(r) -}$ in $B \underline{\otimes}
\ldots \underline{\otimes} B \ $, the elements $b_{j}^{+}$ and
$b_{i}^{-}$ have to interchange their order according to the
braiding \eqref{symmbraid}:
$$
...\otimes \Psi_{B,B} \otimes ...(... \otimes b_{j}^{+} \otimes
b_{i}^{-} \otimes ...) = - ... \otimes b_{i}^{-} \otimes b_{j}^{+}
\otimes ...
$$
In the above relation: $|b_{i}^{-}| = |b_{j}^{+}| = 1 \ $, since
the elements $b_{i}^{-}$ and $b_{j}^{+}$ are odd elements in the
$\mathbb{Z}_{2}$-gradation of the bosonic algebra $B$. Addition of
the above mentioned equations  by parts  produces:
\begin{equation} \label{BrGrAns21}
\{ b_{i}^{(r) -}, b_{j}^{(s) +} \} = 0
\end{equation}
In exactly the same way we compute:
\begin{equation} \label{BrGrAns22}
\{ b_{i}^{(r) -}, b_{j}^{(s) -} \} = \{ b_{i}^{(r) +}, b_{j}^{(s)
+} \} = 0
\end{equation}
for all values of $r, s = 1, 2, ... ,p$ and $r \neq s$. \\
Now it is straightforward to check that relations
\eqref{BrGrAns11} and \eqref{BrGrAns12} coincide with relations
\eqref{GrAns1} of \seref{GrAns}, under the mapping $b_{i}^{(r)
\pm} \leftrightarrow b_{i}^{\alpha \ \pm}$ for $\ r = \alpha$,
while at the same time relations \eqref{BrGrAns21} and
\eqref{BrGrAns22} coincide with relations \eqref{GrAns2} of
\seref{GrAns}, under the same mapping for $\ r = \alpha$ and $\ s
= \beta$. \\
This completes the proof.
\end{proof}
We should underline here that the above mentioned isomorphism is
actually much different than Green's original idea \cite{Green}
about the nature of his ansatz:  By that time the idea of
$\mathbb{Z}_{2}$-graded algebras and $\mathbb{Z}_{2}$-graded
tensor products was a well known fact for mathematicians (see
\cite{Che}) -although rather as a special possibility of forming
superalgebras and tensor products than as an example of the more
general conceptual framework of the braiding- but it was not a
mainstream idea in physics. Green conjectured that his ansatz
should be a subalgebra of the usual tensor product algebra between
$p$-copies of the bosonic algebra $B$ and a copy of a
Clifford-like algebra. For a formulation of Green's original
interpretation one should also see \cite{Rob}. \\

The above proposition permits us a braided reinterpretetion of the
Green's ansatz device. Let us start with a restatement of the
associative superalgebra homomorphism stated in equation
\eqref{GrAnsRepr} of \seref{GrAns}: We start by inductively
defining: $\underline{\Delta}^{(1)} = \underline{\Delta}$ where
$\underline{\Delta}: P_{B} \rightarrow P_{B} \underline{\otimes}
P_{B} \ \ $ is the coproduct of the Parabosonic super-Hopf algebra
$P_{B}$ stated in \prref{superHopfPb} and:
\begin{equation} \label{tensprsupcomult}
\underline{\Delta}^{(p)} = (\underline{\Delta} \otimes Id \otimes
... \otimes Id) \circ \underline{\Delta}^{(p-1)}
\end{equation}
for any integer $\ p \geq 2$. $Id: P_{B} \rightarrow P_{B}$ is the
identity map and there are $p-1$ copies of it on the tensor
product map of the right hand side of the above equation. It is
obvious that $\underline{\Delta}^{(p)}: P_{B} \rightarrow P_{B}
\underline{\otimes} ... \underline{\otimes} P_{B}$ ($P_{B}$
appearing $p+1$ times) is an associative superalgebra
homomorphism. Furthermore we can inductively prove the following
Lemma:
\begin{lemma} \lelabel{Propertensprocomulti}
For any $p \geq 2$ and for any $i = 0, 1, ..., p-1$ we have the
following:
$$
\underline{\Delta}^{(p)} = (Id^{i} \otimes \underline{\Delta}
\otimes Id^{p-i-1}) \circ \underline{\Delta}^{(p-1)}
$$
\end{lemma}
What the above lemma implies is that the definition of
$\underline{\Delta}^{(p)}$ stated in equation
\eqref{tensprsupcomult} is actually independent of the position
where $\underline{\Delta}$ is placed in  the tensor product of the
right hand side. Now we can straightforwardly compute:
\begin{equation} \label{comultparabp}
\underline{\Delta}^{(p-1)}(B_{i}^{\pm}) = \sum_{r=1}^{p} I \otimes
... \otimes I \otimes B_{i}^{\pm} \otimes I \otimes I \otimes ...
\otimes I
\end{equation}
where there are $p$-factors in the tensor product of the right
hand side and the generator $B_{i}^{\pm}$ of $P_{B}$ is placed in
the $r$-th place for $r = 1, 2, ..., p$ and $i = 1, 2, ... \ $.
\begin{corollary} \colabel{braidedFocklike}
Under the isomorphism of \prref{braidGrAns}, the homomorphism
$\pi_{p} : P_{B} \rightarrow G(p)$ stated in equation
\eqref{GrAnsRepr} is actually given by:
\begin{equation} \label{GrAnsreprtenspr}
\pi_{p} = (\phi \otimes ... \otimes \phi) \circ
\underline{\Delta}^{(p-1)} = (\pi_{1} \otimes ... \otimes \pi_{1})
\circ \underline{\Delta}^{(p-1)}
\end{equation}
where $p \geq 2$, $ \ \phi$ is the replacement map defined in
\prref{parabosonstobosons} and there are $p$-copies of it on the
right hand side of the above formula.
\end{corollary}
If we consider the bosonic Fock representation ${}_{B}F$, this is
a $\mathbb{Z}_{2}$-graded complex vector space and at the same
time a $\mathbb{Z}_{2}$-graded $B$-module (we can equivalently
say: a braided $B$-module, where the braiding is given by
\eqref{symmbraid}). Applying the discussion at the end of
\seref{11} we immediately get that the tensor product vector space
${}_{B}F \otimes ... \otimes {}_{B}F$ ($p$-copies) becomes a
braided ($\mathbb{Z}_{2}$-graded) module over the braided tensor
product algebra $B \underline{\otimes} ... \underline{\otimes} B$
($p$-copies). We can straightforwardly check that it is a cyclic
module generated by any of it's elements, thus a simple module (an
irreducible representation). Under the isomorphism of
\prref{braidGrAns} this irreducible, braided,
($\mathbb{Z}_{2}$-graded) $B \underline{\otimes} ...
\underline{\otimes} B$-module coincides with the $G(p)$-module
used by Greenberg and Messiah and specified by equations
\eqref{GrAnsvac}, \eqref{GrAnsvac1}. Furthermore, the
$\mathbb{Z}_{2}$-graded complex vector space ${}_{B}F \otimes ...
\otimes {}_{B}F$ readily becomes a braided
($\mathbb{Z}_{2}$-graded) module over the braided tensor product
algebra $P_{B} \underline{\otimes} ... \underline{\otimes} P_{B}$
($p$-copies) through the superalgebra homomorphism:
\begin{equation} \label{tensproff}
\phi \otimes ... \otimes \phi : P_{B} \underline{\otimes} ...
\underline{\otimes} P_{B} \rightarrow B \underline{\otimes} ...
\underline{\otimes} B
\end{equation}
We note here that since $\phi$ is an epimorphism of superalgebras
(see \prref{parabosonstobosons}) $\phi \otimes ... \otimes \phi$
is also an epimorphism of superalgebras. This in turn implies that
since ${}_{B}F \otimes ... \otimes {}_{B}F$ is an irreducible $B
\underline{\otimes} ... \underline{\otimes} B$ module, it is also
an irreducible $P_{B} \underline{\otimes} ... \underline{\otimes}
P_{B}$ module (the situation being analogous to the one described
in \coref{Parabreprfrombosrepr}).

Finally, the space ${}_{B}F \otimes ... \otimes {}_{B}F$ becomes a
braided ($\mathbb{Z}_{2}$-graded) module over the super-Hopf
algebra $P_{B}$ via the superalgebra homomorphism:
\begin{equation} \label{pcomult}
\underline{\Delta}^{(p-1)} : P_{B} \rightarrow P_{B}
\underline{\otimes} ... \underline{\otimes} P_{B}
\end{equation}
This is a reducible representation of the parabosonic algebra
$P_{B}$, which obviously coincides with the tensor product of
$p$-copies of the braided $P_{B}$-module ${}_{P_{B}}F$, where the
tensor product is constructed according to the method described in
the end of \seref{2}.

What the above discussion finally implies is that Greenberg's and
Messiah's use of the Green's ansatz which led to the
classification of the Fock-like representations ${}_{P_{B}}F(p)$
of the parabosonic algebra (see \thref{FocklikereprPb}) was in
fact nothing else than a systematic study of the reduction of the
tensor product representations of the parabosonic algebra. This is
summarized in the following corollary:
\begin{corollary} \colabel{tensprred}
For every specific value of $p = 1, 2, ... \ $, the vector space
$$
{}_{P_{B}}F \otimes ... \otimes {}_{P_{B}}F
$$
($p$-copies of ${}_{P_{B}}F$), becomes a braided
($\mathbb{Z}_{2}$-graded) $P_{B}$-module through the iterated
comultiplication \eqref{pcomult}. This is a reducible
$P_{B}$-module. It contains an irreducible braided
($\mathbb{Z}_{2}$-graded) $P_{B}$-submodule, which is generated
(as a submodule) by the vacuum vector $|0> \otimes ... \otimes
|0>$. This submodule is exactly the parabosonic Fock-like
representation ${}_{P_{B}}F(p)$.
\end{corollary}
For a similar discussion, but for the case of the finite degree's
of freedom i.e. for the algebra $P_{B}^{(n)}$, see also
\cite{Pal3}, \cite{AnPop}.

Let us note here, that the systematic study of the decomposition
of the ${}_{P_{B}}F \otimes ... \otimes {}_{P_{B}}F \ $
$P_{B}$-module, i.e. the investigation of it's structure and any
other submodules it may contain, is an unsolved problem until
nowadays.

\section{Self-contained sets and generalizations of Green's
ansatz} \selabel{Generalizations}

As we have already explained in \seref{GrAns},
\thref{FocklikereprPb}, Greenberg's and Messiah's study of the
parabosonic Fock-like representations concluded to the fact that
for the parabosonic algebra $P_{B}$, the requirements
$$
\begin{array}{cccc}
  B_{i}^{-} |0> = 0 &  &  & (B_{i}^{-})^{\dagger} = B_{i}^{+} \\
\end{array}
$$
for all $i =1, 2, ... \ $, specify -up to isomorphism- an infinite
collection of irreducible representations, labelled by a positive
integer $p$ and denoted by ${}_{P_{B}}F(p)$. Each of them is
further specified by the condition
$$
B_{i}^{-}B_{j}^{+}|0> = p \delta_{ij} |0>
$$
for all $i, j = 1, 2, ... \ $ and for any positive, integer value
of $p$. Each one of these representations is isomorphic to the
corresponding one constructed -for the same value of $p$- through
the homomorphism: $\pi_{p} : P_{B} \rightarrow G(p)$ described by
equation \eqref{GrAnsRepr} or equivalently
\coref{braidedFocklike}.

In \cite{OhKa}, Ohnuki and Kamefuchi virtually inverted the
problem and posed the following question: Instead of looking for
representations of the $P_{B}$ algebra constructed through
homomorphisms to the Green's ansatz  algebras $G(p)$, they looked
for algebras $\Gamma(p)$ (where $p$ is a positive integer)
satisfying the following conditions: \label{conditions}
\begin{itemize}
\item[(1).] For each specific value of $p$, the corresponding
algebra $\Gamma(p)$ must be generated by a set of generators
denoted by $b_{i}^{\pm}$ ($i = 1, 2, ... \ $), satisfying all the
relations of the parabosonic algebra (among any other relations).

\item[(2).] Every $\Gamma(p)$ algebra must possess the following
property: Any $\Gamma(p)$-module subject for any $i,j = 1, 2, ...
$ to the relation
\begin{equation}
   b_{i}^{-}b_{j}^{+}|0> = q \delta_{ij} |0>
\end{equation}
for some unique vector $|0>$ of the representation space, with $q$
a complex number, must imply: $q = p$.
\end{itemize}
The first of the above conditions virtually means that each one of
the $\Gamma(p)$ algebras (for every positive integer value of $p$)
must be a quotient of the parabosonic algebra $P_{B}$, while the
second condition means that the order of the representation of the
parabosonic algebra $P_{B}$ must be ``absorbed" into the algebra
$\Gamma(p)$. The last statement is the reason that the $\Gamma(p)$
algebras (according to our notation) were named in \cite{OhKa} as
``\emph{self-contained sets of commutation relations}" or
``\emph{commutation relations specific to given order $p$}". In
other words: we are looking for an infinite family of quotients of
the parabosonic algebra $P_{B}$, labelled by a positive integer
$p$, each one of them admitting only one Fock-like representation.

Ohnuki and Kamefuchi, solved the problem for small positive
integer values of $p$ for the parafermionic algebra. In the case
of the parabosonic algebra, they only solved it for $p = 1$
concluding that $\Gamma(1) = B \ $ is no other than the familiar
bosonic algebra. They also conjectured the solution for the values
$p = 2$ and $p = 3$ in the case of the parabosonic algebra. We
state here their conjecture for the case of the $\Gamma(2)$
algebra: It is specified in terms of generators and relations by
the following:
\begin{equation} \label{Gamma2}
\begin{array}{cc}
  \langle b_{k}^{-}, b_{l}^{+}, b_{m}^{-} \rangle_{-}  = 2\delta_{kl}b_{m}^{-} -
  2\delta_{lm}b_{k}^{-},
  &
  \langle b_{m}^{+}, b_{l}^{-}, b_{k}^{+} \rangle_{-}  = 2\delta_{kl}b_{m}^{+} - 2\delta_{lm}b_{k}^{+} \\
   & \\
    \langle b_{k}^{-}, b_{l}^{-}, b_{m}^{+} \rangle_{-} =
    2\delta_{lm}b_{k}^{-}, &
  \langle b_{m}^{-}, b_{l}^{+}, b_{k}^{+} \rangle_{-} = 2\delta_{lm}b_{k}^{+}  \\
   & \\
 \langle b_{k}^{-}, b_{l}^{-}, b_{m}^{-} \rangle_{-} = 0, & \langle b_{k}^{+}, b_{l}^{+}, b_{m}^{+} \rangle_{-} = 0  \\
\end{array}
\end{equation}
for all values of $k, l, m = 1, 2, ... \ $ and $\langle A, B, C
\rangle_{-}$ stands for $ABC-CBA$. \\

We are now going to prove that the $\Gamma(2)$ algebra, specified
in terms of generators and relations by \eqref{Gamma2}, is
actually the solution to the Ohnuki's-Kamefuchi's question, as
this is specified by conditions $(1)$ and $(2)$ stated above. In
other words we are going to prove that the $\Gamma(2)$ algebra is
the \emph{self contained set of commutation relations for $p=2$}
or the \emph{commutation relations specific to given order $p=2$}.
\begin{proposition} \prlabel{selfcontainedsets}
The $\Gamma(2)$ algebra specified in terms of generators and
relations by \eqref{Gamma2}, is a $\mathbb{Z}_{2}$-graded algebra
with it's generators $b_{i}^{\pm} \ $ ($i = 1, 2, ...$), being odd
elements. The $\Gamma(2)$ algebra is a quotient algebra of the
parabosonic algebra $P_{B}$. The ``replacement" map $f_{2}: P_{B}
\rightarrow \Gamma(2)$ defined by: $f_{2}(B_{i}^{\pm}) =
b_{i}^{\pm}$ is a $\mathbb{Z}_{2}$-graded algebra epimorphism
(i.e.: an even algebra epimorphism).
\end{proposition}
\begin{proof}
Let us consider the $\Gamma(2)$ algebra as a quotient of the
tensor algebra $T(V_{X})$ via it's ideal generated by the elements
specified by relations \eqref{Gamma2}. Let us denote
$I_{\Gamma(2)}$ the corresponding ideal of $T(V_{X})$ and
$\pi_{\Gamma(2)}$ the natural projection:
$$
\pi_{\Gamma(2)}:T(V_{X}) \rightarrow \Gamma(2) = T(V_{X}) \Big/
I_{\Gamma(2)}
$$
It is immediate that since relations \eqref{Gamma2} are
homogeneous relations, the $I_{\Gamma(2)}$ ideal is generated by
homogeneous elements, thus $\Gamma(2)$ is a
$\mathbb{Z}_{2}$-graded algebra with it's generators being odd
elements. We can now straightforwardly check that the generators
of the $\Gamma(2)$ superalgebra satisfy the paraboson relations,
i.e.:
\begin{equation}
\begin{array}{c}
  \pi_{\Gamma(2)}( \big[ \{ X_{i}^{\xi},  X_{j}^{\eta}\},
X_{k}^{\epsilon}  \big] -
 (\epsilon - \eta)\delta_{jk}X_{i}^{\xi} - (\epsilon -
 \xi)\delta_{ik}X_{j}^{\eta} ) = \\
              \\
 = \big[ \{ b_{i}^{\xi},  b_{j}^{\eta}\},
b_{k}^{\epsilon}  \big] -
 (\epsilon - \eta)\delta_{jk}b_{i}^{\xi} - (\epsilon -
 \xi)\delta_{ik}b_{j}^{\eta} = 0 \\
\end{array}
\end{equation}
In other words, the above computation shows that, given generators
satisfying relations \eqref{Gamma2}, the same generators have to
satisfy the relations of the parabosonic algebra $P_{B}$ also. In
a more mathematical statement we can equivalently say that
$ker(\pi_{P_{B}}) \subseteq ker(\pi_{\Gamma(2)})$ or equivalently:
$I_{P_{B}} \subseteq I_{\Gamma(2)}$. The rest of the proof flows
exactly as in the case of \prref{parabosonstobosons}. $I_{P_{B}}
\subseteq I_{\Gamma(2)}$ implies that $\pi_{\Gamma(2)}$ is
uniquely extended to an even algebra epimorphism $f_{2}: P_{B}
\rightarrow \Gamma(2)$ through the commutative diagram:
\begin{displaymath}
\xymatrix{T(V_{X}) \ar[rr]^{\pi_{\Gamma(2)}}
\ar[dr]_{\pi_{P_{B}}} & & \Gamma(2) \\
 & P_{B} \ar@{.>}[ur]_{\exists ! \ f_{2}} & }
\end{displaymath}
where $f_{2}$ is completely determined by it's values on the
generators $B_{i}^{\pm}$ of $P_{B}$, i.e.: $f_{2}(B_{i}^{\pm}) =
b_{i}^{\pm}$. We also mention that: $kerf_{2} = I_{\Gamma(2)} /
I_{P_{B}} = \pi_{P_{B}}(I_{\Gamma(2)})$ and we finally have the
$\mathbb{Z}_{2}$-graded algebra isomorphism:
\begin{equation}
\Gamma(2) \cong P_{B} \Big/ ker f_{2}
\end{equation}
\end{proof}
The confirmation of the second of the conditions mentioned in page
\pageref{conditions}, is an easy exercise and can be found in
chapter $5$ of \cite{OhKa}.

Let us stress here, that according to the preceding discussion,
the $\Gamma(p)$ algebras arise as generalizations of the bosonic
algebra rather than the parabosonic algebra. The $\Gamma(p)$
algebra, admits at most one Fock-like representation for the
specific value of $p \ $, just as in the case of the bosonic
algebra there is only one Fock representation corresponding to $p
= 1$ and
described in \seref{bose}. \\

Based on the above interpretation of the $\Gamma(p)$ algebras we
can proceed in constructing a straightforward generalization of
the Green's ansatz device: Inspired by \prref{braidGrAns} and
\coref{braidedFocklike} we can carry the corresponding
construction using now the $\Gamma(2)$ algebra -described in
\prref{selfcontainedsets}- and the $f_{2}$ epimorphism, instead of
the bosonic algebra $B$ and the replacement map $\phi$ used in the
mainstream idea of the Green's ansatz device. Our generalization
of the Green's ansatz consists of the algebra:
\begin{equation} \label{Greengener}
\Gamma(2) \underline{\otimes} \Gamma(2) \underline{\otimes} ...
\underline{\otimes} \Gamma(2)
\end{equation}
where there are $p$ copies of $\Gamma(2)$ in the braided
($\mathbb{Z}_{2}$-graded) tensor product and the homomorphism:
\begin{equation} \label{Greengenerhom}
(f_{2} \otimes ... \otimes f_{2}) \circ
\underline{\Delta}^{(p-1)}: P_{B} \rightarrow \Gamma(2)
\underline{\otimes} \Gamma(2) \underline{\otimes} ...
\underline{\otimes} \Gamma(2)
\end{equation}
where $p \geq 2$, $ \ f_{2}$ is the projection epimorphism defined
in \prref{selfcontainedsets} and there are $p$-copies of it on the
above formula.

We are now going to describe the braided ($\mathbb{Z}_{2}$-graded)
tensor product algebra \eqref{Greengener} in terms of generators
and relations. Let us use -for the elements of $\Gamma(2)
\underline{\otimes} \ldots \underline{\otimes} \Gamma(2) \ $- a
notation analogous to the one used in the proof of
\prref{braidGrAns}, i.e. we denote the elements of
\eqref{Greengener} as:
\begin{equation} \label{genandrelforGreengener}
I \otimes ... \otimes I \otimes b_{i}^{\pm} \otimes I \otimes I
\otimes ... \otimes I \equiv b_{i}^{(r) \pm}
\end{equation}
where in the left hand side of the above there are $p$-factors,
$b_{i}^{\pm} \in \Gamma(2)$ is placed on the $r$-th place and the
identity element $I$ of the $\Gamma(2)$ algebra is everywhere
else, for all values of $i = 1, 2, ... \ $ and $r = 1, 2, ... , p
\ $. The elements $b_{i}^{(r) \pm}$ together with the identity $I
\otimes I \otimes ... \otimes I$ constitute a set of generators of
the braided tensor product algebra $\Gamma(2) \underline{\otimes}
\ldots \underline{\otimes} \Gamma(2) \ $ ($p$-copies of
$\Gamma(2)$). Using the notation specified in
\eqref{genandrelforGreengener} and the braiding \eqref{symmbraid}
of \seref{1}, we have the following relations:
\begin{equation} \label{Greengenerintermsofgenerandrel}
\begin{array}{cc}
  \langle b_{k}^{(r)-}, b_{l}^{(r)+}, b_{m}^{(r)-} \rangle_{-}  = 2\delta_{kl}b_{m}^{(r)-} -
  2\delta_{lm}b_{k}^{(r)-},
  & \langle b_{k}^{(r)-}, b_{l}^{(r)-}, b_{m}^{(r)+} \rangle_{-} =
    2\delta_{lm}b_{k}^{(r)-}   \\
   & \\
  \langle b_{m}^{(r)+}, b_{l}^{(r)-}, b_{k}^{(r)+} \rangle_{-}  = 2\delta_{kl}b_{m}^{(r)+} - 2\delta_{lm}b_{k}^{(r)+}, &
  \langle b_{m}^{(r)-}, b_{l}^{(r)+}, b_{k}^{(r)+} \rangle_{-} = 2\delta_{lm}b_{k}^{(r)+}  \\
   & \\
 \langle b_{k}^{(r)-}, b_{l}^{(r)-}, b_{m}^{(r)-} \rangle_{-} = 0, & \langle b_{k}^{(r)+}, b_{l}^{(r)+}, b_{m}^{(r)+}
 \rangle_{-} = 0  \\
   &   \\
 \{ b_{k}^{(r) -}, b_{l}^{(s) -} \} = \{ b_{k}^{(r) +}, b_{l}^{(s)
+} \} = 0 &  \{ b_{k}^{(r) -}, b_{l}^{(s) +} \} = 0  \\
\end{array}
\end{equation}
which for all values of $k, l = 1, 2, ... \ $, and $r, s = 1, 2,
... p \ $ with $r \neq s$ completely specify the braided
($\mathbb{Z}_{2}$-graded) tensor product algebra
\eqref{Greengener} in terms of generators and relations. It is
worth noting that we have constructed an algebra which mixes both
bilinear and trilinear relations. Finally we can immediately
conclude that the homomorphism \eqref{Greengenerhom} has the same
functional form as in the usual Green ansatz
\begin{equation}
(f_{2} \otimes ... \otimes f_{2}) \circ
\underline{\Delta}^{(p-1)}(B_{i}^{\pm}) = \sum_{\alpha = 1}^{p}
b_{i}^{\alpha \ \pm}
\end{equation}
but the elements on the right hand side of the above are no more
``anticommuting" bosons -as in the usual Green's ansatz case- but
elements of the $\Gamma(2) \underline{\otimes} \ldots
\underline{\otimes} \Gamma(2) \ $ algebra satisfying relations
\eqref{Greengenerintermsofgenerandrel}.

\section{Discussion}
The explicit construction of the parabosonic Fock-like
representation ${}_{P_{B}}F(p)$ for $p > 1$ is an unsolved problem
until nowadays. The problem is yet unsolved even for the case of
the finite degrees of freedom, i.e. even for the algebra
$P_{B}^{(n)}$ for $n > 1$ and $p > 1$ there are no explicit
expressions in the bibliography for the matrix elements of the
${}_{P_{B}^{(n)}}F(p)$ module. Even the -simpler- problems of the
construction of an orthonormal basis for the ${}_{P_{B}}F(p)$ or
the ${}_{P_{B}^{(n)}}F(p)$ modules are open for the case $p > 1$
and $n > 1$. We should note here that a very interesting paper has
appeared \cite{LiStVdJ} -by the time the writing of this work was
almost completed- where some of the above problems are dealt with,
for the case of finite degrees of freedom: The authors consider
the module ${}_{P_{B}^{(n)}}F(p)$ for $n > 1$ and $p > 1$. They
use techniques of induced representations, the well known
\cite{Pal} isomorphism of $P_{B}^{(n)}$ with the Lie superalgebra
$osp(1/2n)$ and elements from the representation theory of
$osp(1/2n)$. They construct an orthogonal Gelfand-Zetlin basis of
${}_{P_{B}^{(n)}}F(p)$ and they calculate explicitly the
corresponding matrix elements. However, their techniques do not
give an answer for the general case of the ${}_{P_{B}}F(p)$ module
with $p > 1$ since the representation theory of the infinite
dimensional Lie superalgebras is yet an unexplored subject. On the
other hand, the braided interpretation of the Green's ansatz
device presented here indicates a possible way for a general
solution of the above mentioned problem. Our approach indicates
that the Green's ansatz algebras $G(p)$ for $p = 1, 2, ... \ $,
should be ``utilized" as $\mathbb{Z}_{2}$-graded algebras (with
their generators being odd elements). Their braided
($\mathbb{Z}_{2}$-graded), irreducible modules have been shown to
give rise to braided ($\mathbb{Z}_{2}$-graded), reducible tensor
product modules of the parabosonic algebra $P_{B}$. The role of
the super-Hopf structure of the $P_{B}$ algebra is essential in
this process. These $P_{B}$-modules finally must be reduced to
their irreducible constituents. The isomorphism of proposition
\prref{braidGrAns} provides us with an analytic tool for such a
calculation to be performed. It will be a very interesting thing
to proceed with the explicit computations, to compare the results
-for the case of the finite degrees of freedom- with the
corresponding results of \cite{LiStVdJ} and to extract the matrix
elements for the general case of the $P_{B}$ algebra with an
infinite number of parabosons. This will be the subject of a
forthcoming work.

Furthermore, we stress here that the above mentioned approach
admits straightforward generalisations for the case of algebras
which describe mixed systems of paraparticles such as the relative
parabose or the relative parafermi sets (see \cite{GreeMe} for
their description and for generalized versions of the Green's
ansatz for these algebras): The relative parafermi set has been
shown to be a $\mathbb{Z}_{2}$-graded algebra \cite{Pal6} while
the relative parabose set has been shown to be a $\mathbb{Z}_{2}
\times \mathbb{Z}_{2}$-graded \cite{Ya} algebra. It would thus be
an interesting idea to apply similar techniques to these algebras
and obtain results about their graded Hopf structure, their
braided representations and their braided tensor products. Such
results, combined with suitable generalizations of
\prref{braidGrAns} can lead us to the explicit construction
(matrix elements, orthonormal basis, character formulas) of
Fock-like modules for mixed parafields. Of course such questions
inevitably involve questions of pure mathematical interest, such
as the possible quasitriangular structures (and thus the possible
braidings) for the $\mathbb{C}(\mathbb{Z}_{2} \times
\mathbb{Z}_{2})$ group Hopf algebra, which up to our knowledge
have not yet been solved. (see \cite{Scheu1} for a relevant
discussion).

Let us close this discussion with some comments regarding the
generalization of Green's ansatz presented in
\seref{Generalizations}. An obvious question which arises at first
glance is the study of the representations of the $\Gamma(2)
\underline{\otimes} ... \underline{\otimes} \Gamma(2)$
superalgebras and specifically the determination of whether they
can help us ``build" -through the homomorphism
\eqref{Greengenerhom}- essentially new representations of the
parabosonic algebra. Of course the relation -if any- of the
$\Gamma(2)$ superalgebra with the $B \underline{\otimes} B$ or
even the $End({}_{B}F \otimes {}_{B}F)$ superalgebras (or any of
it's subalgebras) plays an essential part in answering this
question and should be the starting point of such an
investigation. More generally, we have already mentioned that the
general solution to the problem of the construction of the
\emph{commutation relations specific to given order $p$} or
equivalently the determination of the $\Gamma(p)$ algebra for an
arbitrary positive integer value of $p$ is yet an unsolved
problem. In \cite{OhKa} solutions to the problem are conjectured
for the cases of $p = 2, 3$. But a general solution to this
problem does not exist in the bibliography. Provided such a
general solution, it is easy to see that our method will then
provide a whole family of generalizations to the Green's ansatz.
They will be of the form:
$$
\Gamma(p) \underline{\otimes} ... \underline{\otimes} \Gamma(p)
$$
with $q$ factors of $\Gamma(p)$ appearing on the above braided
($\mathbb{Z}_{2}$-graded) tensor product (Note that in this case
$q$ is an arbitrary positive integer irrelevant of the value of
$p$). Homomorphisms of the form
$$
(f_{p} \otimes ... \otimes f_{p}) \circ
\underline{\Delta}^{(q-1)}: P_{B} \rightarrow \Gamma(p)
\underline{\otimes} ... \underline{\otimes} \Gamma(p)
$$
will provide the suitable link between the representations
theories of $\Gamma(p) \underline{\otimes} ... \underline{\otimes}
\Gamma(p)$ and the parabosonic algebras as long as suitable
generalizations of the \prref{selfcontainedsets} are stated.


\begin{thebibliography}{99}

\bibitem{AnPop} B. Aneva, T. Popov, ``\emph{Hopf structure and Green ansatz of deformed
parastatistics algebras}", J. Phys. A: Math. Gen., v.\textbf{38},
(2005), p.6473-6484.

\bibitem{Che} Chevalley, ``Algebraic Theory of Spinors and Clifford
 Algebras", collected works, v.\textbf{2}, Springer, 1997

\bibitem{CohMon} M. Cohen, S. Montgomery, ``\emph{Group graded rings, smash
products and group actions}", Trans. of the Amer. Math. Soc.,
v.\textbf{282}, 1, (1984), p.237-258

\bibitem{FockCook} J. M. Cook, ``\emph{The mathematics of second
quantization}", Trans. of the Amer. Math. Soc., v.\textbf{74}, 2,
(1953), p.222-245.

\bibitem{Dri} V. G. Drinfeld, ``Quantum Groups", Proc. Int. Cong.
Math., Berkeley, (1986), p. 789-820

\bibitem{Ehrenf} P. Ehrenfest, Zeits. f. Physik, v.\textbf{45},
(1927), p.455

\bibitem{Pal} A. Ch. Ganchev, T.D. Palev, ``\emph{A Lie superalgebraic
interpretation of the parabose statistics}", J. Math. Phys.,
v.\textbf{21}, 4, (1980), p.797-799

\bibitem{Green} H.S. Green, ``\emph{A generalized method of field
quantization}", Phys. Rev., v.\textbf{90}, 2, (1953), p.270

\bibitem{GreeMe} O.W. Greenberg, A.M.L. Messiah, ``\emph{Selection rules
for parafields and the absence of paraparticles in nature}", Phys.
Rev., v.\textbf{138}, 5B, (1965), p.1155

\bibitem{Heis} W. Heisenberg, Zeits. f. Physik, v.\textbf{43},
(1927), p. 172

\bibitem{KaDa1} K. Kanakoglou, C. Daskaloyannis, ``\emph{Graded structure
and Hopf structures in parabosonic algebra. An alternative
approach  to bosonisation}", at: New Techniques in Hopf Algebras
and Graded Ring Theory, edited by: S. Caenepeel, F. Van Oystaeyen,
preprint: arXiv:0706.2825v1 [math-ph]

\bibitem{Ko} B. Konstant, ``\emph{Graded manifolds, graded Lie theory
and prequantization}", Lect. Notes in Math., v.\textbf{570},
Springer, (1977), p.177-306

\bibitem{LiStVdJ} S. Lievens, N. I. Stoilova, Van der Jeugt,
``\emph{The paraboson Fock space and unitary irreducible
representations of the Lie superalgebra $osp(1/2n)$}",
arXiv:0706.4196v1 [hep-th].

\bibitem{Maj2} S. Majid, ``Foundations of Quantum Group Theory",
Cambridge University Press, 1995

\bibitem{Maj3} S. Majid, ``A quantum groups primer", London
Mathematical Society, Lecture Note Series, 292, Cambridge
University Press, 2002.

\bibitem{MiMo} J. Milnor, J. Moore, ``\emph{On the structure of Hopf
algebras}", Ann. of Math., v.\textbf{81}, (1965), p.211-264.

\bibitem{Mon} S. Montgomery, ``Hopf
algebras and their actions on rings", CBMS, Regional Conference
Series in Mathematics, 82, AMS-NSF, 1993.

\bibitem{OhKa} Y. Ohnuki, S. Kamefuchi, ``Quantum field theory and
parastatistics", University of Tokyo press, Tokyo, Springer, 1982.

\bibitem{OmOhKa} M. Omote, Y. Ohnuki, S. Kamefuchi, ``\emph{Fermi-Bose
similarity}", Prog. Theor. Phys., v.\textbf{56}, 6, (1976),
p.1948-1964.

\bibitem{Pal3} T.D. Palev, ``\emph{Algebraic structure of Green's ansatz
and it's q-deformed analogue}", J. Phys. A: Math. Gen.,
\textbf{27}, (1994), p.7373-7385.

\bibitem{Pal6} T.D. Palev, ``\emph{Parabose and parafermi operators as
generators of orthosymplectic Lie superalgebras}", J. Math. Phys.
\textbf{23}, 6, (1982), p.1100-1102.

\bibitem{Put} C. R. Putnam, ``Commutation Properties of Hilbert
Space Operators and Related Topics", Springer-Verlag, Berlin,
(1967).

\bibitem{Ques} C. Quesne, ``\emph{Interpretation and extension of
Green's ansatz for paraparticles}", Phys. Lett. A, \textbf{260},
(1999), p.437-440.

\bibitem{ReSi} M. Reed and B. Simon, ``Methods of Modern
Mathematical Physics", vol.I, Academic Press (1975).

\bibitem{RiNa} F. Riesz and B. Sz. Nagy, ``Functional Analysis",
Frederick Ungar Pub. Co., N.Y. (1955).

\bibitem{Rob} S. Robbins, ``\emph{The uniqueness of the one dimensional
paraboson field}", Trans. of the Amer. Math. Soc., v.\textbf{209},
(1975), p.389-397 / ``\emph{Paraboson uniqueness for infinitely
many degrees of freedom}", J. of Math. Physics, v.\textbf{18}, 5,
(1977), p.997-1005.

\bibitem{Scheu} M. \ Scheunert, ``\emph{The theory of Lie superalgebras}",
 Lect. Not. Math., v.\textbf{716}, Springer, (1978), p.1-270.

\bibitem{Scheu1} M. \ Scheunert, ``\emph{Universal $R$-matrices
for finite abelian groups - a new look at graded multilinear
algebra}", arXiv:q-alg/9508016v1

\bibitem{Schroed} D. Schr\"{o}edinger, Ann. d. Physik,
v.\textbf{79}, p.361, p489, v.\textbf{80}, p.437, v.\textbf{81},
p.109 (1926).

\bibitem{Seg} I. E. Segal, ``\emph{Tensor algebras over Hilbert spaces
I}", Trans. of the Amer. Math. Soc., v.\textbf{81}, 1, (1956),
p.106-134 / ``\emph{Tensor algebras over Hilbert spaces II}",
Annals of Math., v.\textbf{63}, 1, (1956), p.160-175.

\bibitem{Stee} N.E. Steenrod, ``\emph{The cohomology algebra of a space}",
 Enseign. Math. II., ser.\textbf{7}, (1961), p.153-178.

\bibitem{Vo} D. V. Volkov, ``\emph{On the quantization of
half-integer spin fields}", Sov. Phys.-JETP \textbf{9}, (1959) p.
1107-1111 / ``\emph{S-matrix in the generalized quantization
method}", Sov. Phys.-JETP \textbf{11}, (1960), p.375-378.

\bibitem{Wint} A. Wintner, ``\emph{The unboundedness of quantum mechanical
matrices}", Phys. Rev., v.\textbf{71}, (1947), p. 738-739.

\bibitem{Ya} W. Yang, Sicong Jing, ``\emph{Graded Lie algebra generating
of parastatistical algebraic structure}", Commun. in theor. phys.
v.\textbf{36}, 6, (2001), p.647-650 / ``\emph{A new kind of graded
Lie algebra and parastatisical supersymmetry}", Science in China
(Series A), v.\textbf{44}, 9, (2001)

\end{thebibliography}
\end{document}